\documentclass[journal,comsoc]{IEEEtran}
\usepackage{amsfonts,amsmath,amssymb,amsthm}
\usepackage{multirow,graphicx,cite,color}

\newtheorem{Theorem}{Theorem}
\newtheorem{Remark}{Remark}
\newtheorem{Proposition}{Proposition}

\newtheorem{Application}{Application}

\normalsize

\ifCLASSINFOpdf
\else
\fi

\begin{document}
%
\title{Towards Practical Private Information Retrieval from MDS Array Codes}
%
%
%

\author{Jie Li,~\IEEEmembership{Member,~IEEE,}
        David Karpuk,
        and~Camilla Hollanti,~\IEEEmembership{Member,~IEEE}
\thanks{J. Li is with the Department of Mathematics and Systems Analysis,
                    Aalto University, FI-00076 Aalto,  Finland. He is also an adjunct researcher at  Hubei Key Laboratory of Applied Mathematics, Faculty of Mathematics and Statistics, Hubei University,
Wuhan 430062, China (e-mail: jie.0.li@aalto.fi; jieli873@gmail.com).}
\thanks{D. Karpuk was previously with the Departamento de Matem\'aticas, Universidad de los Andes, Bogot\'a, Colombia. He is currently with F-Secure Corporation, Helsinki, Finland (email: davekarpuk@gmail.com).}
\thanks{C. Hollanti is with the Department of Mathematics and Systems Analysis,
                    Aalto University, FI-00076 Aalto,  Finland  (e-mail: camilla.hollanti@aalto.fi).}
\thanks{The work of J. Li was supported in part by the National Science Foundation of China under Grant No. 61801176. The work of C. Hollanti was supported by the Academy of Finland, under Grants No. 303819 and 318937, by the Finnish Academy of Science and Letters, and by the Technical University of Munich, Institute for Advanced Study, funded by the German Excellence Initiative and the EU 7th Framework Programme under Grant Agreement No. 291763, via a Hans Fischer Fellowship.}}

\maketitle

\begin{abstract}
Private information retrieval (PIR)  is the problem of privately retrieving one out of $M$ original files from $N$ severs, \textit{i.e.}, each individual server gains no information on the identity of the file that the user is requesting. Usually, the $M$ files are replicated or encoded by a maximum distance separable (MDS) code and then stored across the $N$ servers. Compared to mere replication,  MDS-coded servers can significantly reduce the storage overhead. Particularly, PIR from minimum storage regenerating (MSR)  coded servers can simultaneously  reduce the repair bandwidth when repairing failed servers. Existing PIR protocols from MSR-coded servers either require large sub-packetization levels or are not capacity-achieving. In this paper, a PIR protocol from MDS array codes is proposed, subsuming PIR from MSR-coded servers as a special case. Particularly, only the case of non-colluding, honest-but-curious servers is considered. The  retrieval rate of the new PIR protocol achieves the capacity of PIR from MDS-/MSR-coded servers. By choosing different MDS array codes, the new PIR protocol can have varying advantages when compared with existing protocols, \textit{e.g.}, 1) small sub-packetization,  2) (near-)optimal repair bandwidth, 3)
implementable over the binary field $\mathbf{F}_2$.
\end{abstract}


%
\IEEEpeerreviewmaketitle

\section{Introduction}
%
%
%
%
\IEEEPARstart{I}{n} private information retrieval, a user wishes to retrieve a file from a database without disclosing the identity of the desired file. This PIR problem was first introduced by Chor \textit{et al.} in 1995 \cite{Chor1995}, and received a lot of attention since then \cite{Sun_capacity_replicated,Sun_replicated_downloadcost,Tian2019,Banawan_MDS_capacity,XU_Zhang_SciCN,Razane-IT18,Freij-SIAM,freij2018t,kumar2017private,Zhu2019,Zhou2019,Dorkson_PM_PIR,Julien_PM_PIR,kumar2019achieving,fazeli2015codes,blackburn2017pir,sun2019breaking}. In the classical PIR model, $N$ servers  each   store  a copy of all the $M$ files. The user sends a query to each server who then responds by sending a response to the user. The protocol should be designed in such a way that the user is then  able to decode the desired file from the responses received without revealing the identity of the desired file to any individual server. This is referred to as  PIR from replicated servers.

One of the most important metrics to evaluate the performance of PIR protocols is the \textit{retrieval rate}, which    is defined as the number of bits that the user can privately retrieve per bit of download
data.
The maximum value of the retrieval rate of a given PIR setting is termed  \textit{PIR capacity}, which was recently characterized by Sun and Jafar \cite{Sun_capacity_replicated} for the case of replicated servers   as
\begin{eqnarray}\label{Eqn_capacity_repli}
C_{replicated}=\left(1+\frac{1}{N}+\frac{1}{N^2}+\ldots+\frac{1}{N^{M-1}}\right)^{-1}.
\end{eqnarray}
In \cite{Sun_capacity_replicated}, a capacity-achieving PIR protocol  was also presented, which requires that the length of each file (also known as sub-packetization in the literature) should be $N^M$, and was reduced to  $N^{M-1}$ later by the same authors in \cite{Sun_replicated_downloadcost}, where it is also proved that $N^{M-1}$  is the optimal   sub-packetization under the assumption that the total  length of the responses  from all the servers (\textit{i.e.,} download cost) is symmetric in different random realizations of queries, \textit{i.e.,} the download cost is identical over all the realizations.
In practice, large sub-packetization is not preferred since   the complexity in practical implementations would be increased.
Recently, Tian \textit{et al.} \cite{Tian2019} proposed a novel PIR protocol for replicated servers with the sub-packetization being $N-1$, while the retrieval rate achieves the capacity in \eqref{Eqn_capacity_repli}. Unlike existing ones in \cite{Sun_capacity_replicated,Sun_replicated_downloadcost}, the new PIR protocol in \cite{Tian2019} allows  the total length of the responses  from the servers   to be asymmetric in different realizations of queries, which is the key to reducing the  sub-packetization.

However, for PIR protocols from replicated servers, a rather excessive storage overhead is required. This concern motivated the study of PIR from $(N, K)$ MDS-coded servers, where each server  stores  $\frac{1}{K}$ of each file.  The capacity of the PIR from $(N, K)$ MDS-coded servers was characterized by Banawan and Ulukus \cite{Banawan_MDS_capacity} as
    \begin{eqnarray}\label{Eqn_capacity_MDS}
    C_{MDS}=\left(1+\frac{K}{N}+\frac{K^2}{N^2}+\ldots+\frac{K^{M-1}}{N^{M-1}}\right)^{-1},
    \end{eqnarray}
where a capacity-achieving protocol was also proposed with the sub-packetization being $KN^M$. Recently,  the sub-packetization of the capacity-achieving PIR from  MDS-coded servers was reduced to $K\left(\frac{N}{\gcd(N, K)}\right)^{M-1}$ by Xu and Zhang in \cite{XU_Zhang_SciCN}, which is also optimal under the assumption that the   length of the responses  is a constant in different realizations of queries. In \cite{Razane-IT18},   Tajeddine and Rouayheb proposed a  PIR protocol from $(N, K)$ MDS-coded servers, where the sub-packetization can be as small as  $\frac{K(N-K)}{\gcd(N, K)}$, with the retrieval rate being $1-\frac{K}{N}$, which is independent of the file number $M$ and  asymptotically achieves the capacity in \eqref{Eqn_capacity_MDS}.
 In \cite{Freij-SIAM}, Freij-Hollanti \textit{et al.} presented a general \emph{star product} PIR protocol for MDS-coded servers with collusion including the work in  \cite{Razane-IT18} as a special case, and also allowing for a low sub-packetization level. Very recently, Zhu \textit{et al.} \cite{Zhu2019} and Zhou \textit{et al.} \cite{Zhou2019} independently discovered  PIR protocols from $(N, K)$ MDS-coded servers with the  sub-packetization being $\frac{K(N-K)}{\gcd(N, K)}$, while the   retrieval rate achieves the capacity in \eqref{Eqn_capacity_MDS}. The key ingredient of these protocols is similar to that in \cite{Tian2019}, \textit{i.e.},  the download cost is asymmetric in different realizations of queries.

In practical systems, each individual server may also suffer from failures \cite{GFS}. In such a scenario, the failed server could be repaired by  introducing a replacement node and downloading   data from some other $D$ surviving servers independently, where $D$ is referred to as the \textit{repair degree} and the total amount of data downloaded is referred to as the \textit{repair bandwidth} in the literature. However,  for most existing PIR protocols from MDS-coded servers, a repair bandwidth equal to the size of all the files is needed, which is rather inefficient. To address this issue, PIR protocols from regenerating codes, which can efficiently repair a single server failure while still maintaining a high retrieval rate were considered recently in  \cite{Dorkson_PM_PIR,Julien_PM_PIR}. \textit{Regenerating codes} are a kind of \emph{vector codes} or \emph{array codes}   that can reduce the repair bandwidth, and were originally  introduced by Dimakis \textit{et al.} in  \cite{Dimakis}, also  characterizing the tradeoff between the storage and repair bandwidth in distributed storage systems. The two most interesting classes of regenerating codes are the \emph{Minimum Storage Regenerating (MSR) codes} and the \emph{Minimum Bandwidth Regenerating (MBR) codes}, which received a lot of attention \cite{PM,hadamard,Barg1,Barg2,Sasidharan-Kumar2,YiLiu,Hadamard_strategy,transform-IT,transform-ISIT,Bitran,Goparaju,MDR2Wang,Sohei_ISIT2018,Soheil_MBR,Tian433,Hashtag,Basic_code,han2015update}.
In regenerating codes, the data stored in each node can be viewed as a vector of length $\alpha$, which is referred to as the \textit{node capacity} in the literature.
The node capacity  $\alpha$ and the repair bandwidth $\gamma$ of  $(N, K)$ MBR codes and $(N, K)$ MSR codes with repair degree $D$ were respectively characterized as
\begin{eqnarray*}
\nonumber \alpha_{\mathbf{MBR}}=\frac{2\mathcal{B}D}{2KD-K^2+K},\ \gamma_{\mathbf{MBR}}=\frac{2\mathcal{B}D}{2KD-K^2+K},
\end{eqnarray*}
\begin{eqnarray}\label{Eqn_MSR_parameters}
 \alpha_{\mathbf{MSR}}=\frac{\mathcal{B}}{K},\ \gamma_{\mathbf{MSR}}=\frac{\mathcal{B}D}{K(D-K+1)},
\end{eqnarray}
where $\mathcal{B}$ is the size of the original files stored across the system.

In \cite{Dorkson_PM_PIR}, a PIR protocol from the  $(N, K)$ product-matrix-MSR (PM-MSR) code in \cite{PM} was proposed, where each failed server can be repaired with theoretically minimum repair bandwidth  achieving \eqref{Eqn_MSR_parameters}.  However, the retrieval rate is $1-\frac{2K-2}{N}$, which is far from the capacity in \eqref{Eqn_capacity_MDS}. The retrieval rate of the PIR protocol from the  $(N, K)$  PM-MSR code was later improved to $1-\frac{4K-2}{3N-2K+4}$ by Lavauzelle \textit{et al.} in \cite{Julien_PM_PIR}, but still can not achieve the  capacity in \eqref{Eqn_capacity_MDS}. Nevertheless, it is worth noting that in \cite{Julien_PM_PIR}, a PIR protocol from $(N, K)$  PM-MBR code was also considered, with a high retrieval rate larger than $1-\frac{K}{N}$ but a slightly larger storage overhead than that of the MDS-coded servers. Very recently, Patra and Kashyap \cite{Patra_capacity_MSR}  showed that the capacity of a PIR protocol from MSR codes is the same as that from  MDS codes, and a capacity-achieving construction was proposed, however, with a large sub-packetization being $\alpha K N^M$, where $\alpha$ is the node capacity of the specific MSR code employed.

Another important metric in PIR is the field size, which directly  affects the system complexity. To the best of our knowledge, nearly all existing PIR protocols from MDS-coded servers are over non-binary fields. There do exist some (non-MDS) PIR protocols over the binary field in the literature, such as the ones in \cite{freij2018t}.

In this paper, we
follow the line of work in \cite{Banawan_MDS_capacity,Razane-IT18,Freij-SIAM}, where the encoding is within each file and hence the system is dynamic in that adding a file is easy and does not affect the rest of the system. This is in contrast to the work in \cite{fazeli2015codes,blackburn2017pir,sun2019breaking}, where the encoding is across different files.
Motivated by the works in \cite{Zhu2019,Zhou2019},  by taking into account  both the sub-packetization and the repair bandwidth, we propose a novel  PIR protocol from MDS array codes, which generalizes the works in \cite{Dorkson_PM_PIR,Julien_PM_PIR,Patra_capacity_MSR}. For simplicity, in this paper, only the case of non-colluding, honest-but-curious servers is considered.   By choosing different MDS array codes, the new PIR protocol has the following advantages A1, A2, and A3 (or A4):
\begin{enumerate}
    \item  [A1.] The  retrieval rate achieves the capacity in \eqref{Eqn_capacity_MDS}, which outperforms the works in \cite{Dorkson_PM_PIR,Julien_PM_PIR}.
    \item [A2.] A failed server can be repaired with a small (or optimal)  repair bandwidth, which outperforms the works in \cite{Banawan_MDS_capacity, Zhu2019,Zhou2019}.
         \item [A3.] The sub-packetization is relatively small, and outperforms that in \cite{Patra_capacity_MSR}.
    \item [A4.] Can be possibly implemented over $\mathbf{F}_2$, which is impossible for  all existing PIR protocols from MDS-coded servers.
\end{enumerate}

Except to the above advantages, the new PIR protocol has a drawback that the download cost is not a constant in  different realizations of queries. Nevertheless, in \cite{Tian2019}, the authors proposed three symmetrization techniques, applying which  on an asymmetric PIR protocol can produce  a new one that is  symmetric, but at the cost of increasing the sub-packetization, and the resultant PIR protocol would not have advantages compared with the capacity-achieving PIR protocols with a constant download size in terms of the sub-packetization. Overall, asymmetric download cost is a key to significantly reduce the sub-packetization of a capacity-achieving PIR protocol.

The rest of this paper is organized as follows. Section II introduces some basic preliminaries of MDS array codes and PIR models. Section III presents the new PIR protocol from MDS array codes, followed by the asserted properties. An illustrative example is also given. Section IV gives four specific applications of the  PIR protocol proposed in Section III, by choosing four specific MDS array codes, and  comparisons among the new PIR protocols and existing ones are also provided. Finally, in Section V we provide concluding remarks.


%

\section{Preliminaries}
Denote by $q$ a prime power and $\mathbf{F}_q$ the finite field with $q$ elements. For any two positive integers $x$ and $y$, denote by $[x, y)$ the set $\{x, x+1, \ldots, y-1\}$. Throughout this paper, we use superscripts to refer to files, subscripts to refer to servers, and parenthetical indices for (block) entries of a vector. For example,
\begin{itemize}
    \item $W^i$ denotes the $i$-th file,
     \item  $Q_j^{\theta}$ denotes the query sent to the $j$-th server while requesting file $\theta$,
     \item $A(t)$ denotes the $t$-th (block) entry of the vector $A$.
\end{itemize}

\subsection{MDS Array Codes}

Scalar MDS codes require significantly large bandwidth during a repair process, while some MDS array codes only require a smaller repair bandwidth. Compared with scalar MDS codes, the codewords of MDS array codes are in 2 dimensions, \textit{i.e.}, can be viewed as a matrix. Examples of MDS array codes are  MSR codes  and their variation-- $\epsilon$-MSR codes \cite{Rawat,Li_eMSR_ISIT}.


 MSR codes are a kind of MDS array codes, which were originally  introduced by Dimakis \textit{et al.} in order to optimally repair a failed node in distributed storage systems \cite{Dimakis}. 
However, a major concern of high-rate MSR codes is that the node capacity is significantly large, the smallest node capacity among all the known constructions is $(n-k)^{\frac{n}{n-k}}$ for general parameters $n$ and $k$ \cite{Barg2,Sasidharan-Kumar2,transform-IT}. To address this concern, a variation of MSR codes named  $\epsilon$-MSR codes \cite{Rawat} has been studied, which  have a significantly smaller node capacity than that of MSR codes but with near-optimal repair bandwidth. Besides $\epsilon$-MSR codes, there are also other kinds of MDS array codes that can  reduce the node capacity, such as the piggybacking design codes \cite{Rashmi_piggyback,Yuan_piggyback}, which are operating on multiple instances of scalar MDS codes.

\subsection{PIR models}\label{sec:PIR}

Suppose there are $M$ files over $\mathbf{F}_q$, denoted by $W^i$, $i\in [0, M)$, and are stored across an $(N, K, \alpha)$ MDS array-coded distributed storage system. These files  are   independent
and identically distributed with
\begin{eqnarray*}
  H(W^i)&=&L,~ i\in [0, M),\\
  H\left(W^0,~W^1,\ldots,W^{M-1}\right)&=&ML,
\end{eqnarray*}
where $H(*)$ denotes the entropy function with base $q$.

For each $i\in [0, M)$,  $W^i$ can be represented in the form of a  $B\times K\alpha$  matrix as
\begin{eqnarray}
 W^i=\left(\begin{array}{cccc}
W^i_{0,0} & W^i_{0,1} & \cdots & W^i_{0,K-1}\\
W^i_{1,0} & W^i_{1,1} & \cdots & W^i_{1,K-1}\\
\vdots & \vdots & \ddots & \vdots\\
W^i_{B-1,0} & W^i_{B-1,1} & \cdots & W^i_{B-1,K-1}\\
\end{array}\right),
\end{eqnarray}
for some positive integer   $B$,
where $W^i_{j,l}$ is a row vector of length $\alpha$ for $j\in [0, B)$ and $l\in [0, K)$, and $(W^i_{j,0},  W^i_{j,1},  \ldots,  W^i_{j,K-1})$ is called the $j$-th sub-stripe of $W^i$. Therefore, $L=\alpha B K$.

{\bf Encoding process: }
Encode each sub-stripe  of each file  by  an $(N, K, \alpha)$ MDS array code\footnote{Unless otherwise stated, we only consider linear code in this paper.}, \textit{i.e.,}
\begin{eqnarray}\label{Eqn_encoding}
\left(Y^i_{j,0},  Y^i_{j,1}, \ldots,  Y^i_{j,N-1}\right)=\left(W^i_{j,0},  W^i_{j,1},  \ldots,  W^i_{j,K-1}\right)G
\end{eqnarray}
for $i\in [0, M),~j\in [0, B)$,
where $G$ is the generator matrix\footnote{Recently, there are several MDS array codes defined in the form of parity-check matrix \cite{Barg1, Barg2,YiLiu}, nevertheless, they can also be converted into the generator matrix form.} of the $(N, K, \alpha)$ MDS array code,  which is a $K \alpha\times N\alpha$ matrix, usually represented as a   $K  \times N $ block matrix with each block entry being an $\alpha \times \alpha$ matrix.


Arrange the $M$ files   into a matrix as
\begin{eqnarray}
W=\left(\begin{array}{c}
W^0 \\
W^1 \\
\vdots\\
W^{M-1}\\
\end{array}\right),
\end{eqnarray}
then the whole encoding process can be expressed as
\begin{eqnarray}
Y=WG=\left(\begin{array}{c}
Y^0 \\
Y^1 \\
\vdots\\
Y^{M-1}\\
\end{array}\right)
\end{eqnarray}
where
\begin{eqnarray*}
Y^i=W^iG=\left(\begin{array}{cccc}
Y^i_{0,0} & Y^i_{0,1} & \cdots & Y^i_{0,N-1}\\
Y^i_{1,0} & Y^i_{1,1} & \cdots & Y^i_{1,N-1}\\
\vdots & \vdots & \ddots & \vdots\\
Y^i_{B-1,0} & Y^i_{B-1,1} & \cdots & Y^i_{B-1,N-1}\\
\end{array}\right).
\end{eqnarray*}

Then the data stored in  server $i$ is
\begin{eqnarray*}
Y[:, i]&=&\left((Y^0_{0,i})^T, (Y^0_{1,i})^T, \ldots, (Y^0_{B-1,i})^T, \ldots, (Y^{M-1}_{0,i})^T,\right.\\ &&\left. (Y^{M-1}_{1,i})^T, \ldots, (Y^{M-1}_{B-1,i})^T\right)^T,~i\in [0, N),
\end{eqnarray*}
where $Y[:, i]$ denotes the $i$-th block column of $Y$.

Suppose that the user wishes to retrieve file $\theta$ privately through a PIR protocol, where $\theta$ is uniformly distributed on   $[0, M)$, then the protocol consists of the following phases:
\begin{itemize}
    \item [(i)] {\bf Query Phase:} The user generates $N$ queries $Q^{(0,\theta)}, \ldots, Q^{(N-1,\theta)}$ according to  some distribution over a  certain probability space and sends $Q^{(i,\theta)}$ to the   server $i$ for $i\in [0, N)$. Indeed,  the queries are generated with no realizations of the files, \textit{i.e.,}
    \begin{equation*}
        I\left(W^0, \ldots, W^{M-1}; Q^{(0,\theta)}, \ldots, Q^{(N-1,\theta)}\right)=0,
    \end{equation*}
where $I(A; B)$ denotes the mutual information between $A$ and $B$.

\item [(ii)] {\bf Response Phase:} Upon receiving the query, for all $i\in [0,N)$, server $i$ returns the response $A^{(i,\theta)}$ to the user. The response is a deterministic function of $Q^{(i,\theta)}$ and the data $Y[:, i]$ stored at server $i$, \textit{i.e.,}
  \begin{equation*}
H\left(A^{(i,\theta)}|Q^{(i,\theta)}, Y[:, i]\right)=0,
  \end{equation*}
where $H\left(X|Y\right)$ denotes the entropy of $X$ conditioning on $Y$.
\end{itemize}

The PIR protocol should be designed  to guarantee
\begin{itemize}
    \item {\bf Correctness:} With the responses and queries, the user can get the desired file $W^{\theta}$, \textit{i.e.,}
    \begin{equation*}
        H\left(W^{\theta}|A^{(0,\theta)},\ldots,A^{(N-1,\theta)},Q^{(0,\theta)},\ldots,Q^{(N-1,\theta)}\right)=0.
    \end{equation*}

    \item {\bf Privacy:} Each server   should learn nothing about which file the user requested, \textit{i.e.,}
    \begin{equation*}
      I\left(\theta;Q^{(i,\theta)},Y[:, i]\right)=0,\, i\in [0, N),
    \end{equation*}
which is equivalent to
    \begin{equation*}
     \mathrm{Prob}\left(Q^{(i,\theta_0)}=Q^*\right)=\mathrm{Prob}\left(Q^{(i,\theta_1)}=Q^*\right)
    \end{equation*}
for every $i\in [0, N)$, $\theta_0, \theta_1\in [0, M)$, and $Q^*$ in the probability space.
\end{itemize}

Formally, the PIR retrieval rate   is defined as
\begin{equation}\label{Eqn_rate}
  R=\frac{L}{D_c},
\end{equation}
where   $D_c$ denotes the expected value   of the amount of data
 downloaded by the user from all the
servers \cite{Sun_capacity_replicated}.

In general, it is preferred that a PIR protocol has the following properties:
\begin{itemize}
  \item High retrieval rate,
  \item Small sub-packetization,
  \item Small finite fields,
  \item Small repair bandwidth when repairing  a failed server.
\end{itemize}

\section{A New PIR Protocol From MDS Array Codes}

In this section, we propose a new PIR protocol from MDS array codes,   the technique proposed in the following can be viewed as a generalization of those in \cite{Zhu2019} and \cite{Zhou2019}.

\subsection{A new PIR protocol from MDS array codes}\label{sec:PIR_cons}

Consider a  distributed storage system storing $M$ files across $N$ servers based on an $(N, K, \alpha)$ MDS array code, as depicted in  Section \ref{sec:PIR}.
Let
\begin{eqnarray}\label{Eqn_bs}
B=\frac{N-K}{\gcd(N, K)},~ S=\frac{K}{\gcd(N, K)}.
\end{eqnarray}
Denote by $\Omega$ the set of $M\times S$ matrices over $[0, B+S)$ with any two entries in any row being distinct, \textit{i.e.,}

\begin{eqnarray}
\nonumber\Omega&=&\{Q=(q_{i,j})_{i\in [0, M),j\in [0, S)}\in [0, B+S)^{M\times S}:\\\label{Eqn_query_space}&& q_{i,j} \ne q_{i,l},
 i\in [0, M),~j,l\in [0,S), j\ne l \}.
\end{eqnarray}

For convenience, define
\begin{eqnarray}\label{Eqn_dummy}
Y^i_{j, l}=\boldsymbol{0}_{\alpha}, \mbox{~for~} i\in [0, M),~  j\in [B, B+S),~ l\in [0, N),
\end{eqnarray}
where $\boldsymbol{0}_{\alpha}$ denotes the zero  row vector of length $\alpha$, and will be abbreviated as $\boldsymbol{0}$ in the sequel if its length is clear.

Now we are ready to propose the query phase and response phase of the new PIR protocol.

{\bf Query Phase:}
We assume that the user wishes to retrieve file $W^\theta$.
The user first randomly chooses a matrix $Q$ in $\Omega$. Next, the user generates the queries based on the randomly chosen matrix $Q$ as
\begin{eqnarray}\label{Eqn_Q_i}
\setlength{\arraycolsep}{0.6pt}\small
Q^{(i,\theta)}\hspace{-0.8mm}=\hspace{-2mm} 
       \left(\hspace{-1mm}\begin{array}{cccc}
q_{0,0} & q_{0,1} & \cdots & q_{0,S-1}\\
\vdots & \vdots & \ddots & \vdots\\
q_{\theta-1,0} & q_{\theta-1,1} & \cdots & q_{\theta-1,S-1}\\
(q_{\theta,0}+i)_{B+S} & (q_{\theta,1}+i)_{B+S} & \cdots & (q_{\theta,S-1}+i)_{B+S}\\
q_{\theta+1,0} & q_{\theta+1,1} & \cdots & q_{\theta+1,S-1}\\
\vdots & \vdots & \ddots & \vdots\\
q_{M-1,0} & q_{M-1,1} & \cdots & q_{M-1,S-1}
\end{array}\hspace{-0.8mm}\right),
\end{eqnarray}
where $i\in [0, N)$, $(*)_{B+S}$ denotes the modulo $B+S$ operation.
Then the query $Q^{(i,\theta)}$ is sent to server $i$ for $i\in [0, N)$.

{\bf Response Phase:} Upon receiving the query, server $i$ sends  the response in the form of a row vector of length  $\alpha\times S$  as

\begin{eqnarray}\label{Eqn_ans_j}
A^{(i,\theta)}=\left(A^{(i,\theta)}(0), A^{(i,\theta)}(1), \ldots,A^{(i,\theta)}(S-1)\right)
\end{eqnarray}
with
\begin{eqnarray}\label{Eqn_ans_j_iteration}
A^{(i,\theta)}(j)=\sum\limits_{l=0,l\ne \theta}^{M-1}Y^l_{q_{l,j},i}+Y^{\theta}_{(q_{\theta,j}+i)_{B+S},i},~j\in [0, S).
\end{eqnarray}

From \eqref{Eqn_Q_i} and \eqref{Eqn_ans_j_iteration}, the following proposition is obvious.
\begin{Proposition}\label{prop1}
If every entry of  $Q^{(i,\theta)}[:, j]$ (\textit{i.e.,} the $j$-th column of $Q^{(i,\theta)}$) is greater than or equal to $B$, then $A^{(i,\theta)}(j)=\boldsymbol{0}$.
\end{Proposition}

According to Proposition \ref{prop1},   the user knows  which block entry in \eqref{Eqn_ans_j} is a zero vector from the queries in \eqref{Eqn_Q_i}. In practice, the user can just delete such  columns   from $Q^{(i,\theta)}$. Therefore,
the length of the response $A^{(i,\theta)}$ from server $i$ is
\begin{eqnarray}\label{Eqn_len}
l_{i}=\alpha\times |\{j:Q^{(i,\theta)}[:, j]\not\in [B, B+S)^M,j\in [0, S)\}|.
\end{eqnarray}

Similar to Fact 3 in \cite{Zhu2019}, we have the following result according to \eqref{Eqn_query_space} and \eqref{Eqn_Q_i}.

\begin{Proposition}\label{prop2}
For any given $j\in [0, S)$, $Q^{(i,\theta)}[:, j]$ is independent and uniformly distributed over $[0, B+S)^M$ for every $i\in [0, N)$ and $\theta\in [0, M)$.
\end{Proposition}

\subsection{Correctness and privacy of the new PIR protocol}

\begin{Theorem}
The   retrieval rate of the PIR protocol  proposed in Section \ref{sec:PIR_cons} meets the capacity in \eqref{Eqn_capacity_MDS}, with sub-packetization being $L=\alpha B K$, where $B$ is defined as in \eqref{Eqn_bs}.
\end{Theorem}

\begin{proof}
The assertion is done by proving the correctness and privacy, and  examining  the  length of the responses.

{\bf Proof of correctness:}
The user now receives
\begin{eqnarray}
\{A^{(i,\theta)}: i\in [0, N)\}=\bigcup_{j=0}^{S-1}\{A^{(i,\theta)}(j): i\in [0, N)\}.
\end{eqnarray}

Note that
$\left(Y^l_{q_{l,j},0}, Y^l_{q_{l,j},1}, \ldots, Y^l_{q_{l,j},N-1}\right)$
is a codeword of the $(N, K)$ MDS array code for all $l\in [0, M)$ and $j\in [0, S)$ by \eqref{Eqn_encoding} and \eqref{Eqn_dummy}, so does
\begin{eqnarray}\label{Eqn_codeword}
\left(\sum\limits_{l=0,l\ne \theta}^{M-1}Y^l_{q_{l,j},0}, \sum\limits_{l=0,l\ne \theta}^{M-1}Y^l_{q_{l,j},1}, \ldots, \sum\limits_{l=0,l\ne \theta}^{M-1}Y^l_{q_{l,j},N-1}\right).
\end{eqnarray}

For any given $j\in [0, S)$, define
\begin{eqnarray*}
  Z_j &=&  \{i\in [0, N):(q_{\theta,j}+i)_{B+S}\in [B, B+S)\},
\end{eqnarray*}

Obviously, $|Z_j|=S\times \gcd(N, K)=K$ since $(q_{\theta,j}+i)_{B+S}$ takes each value in the set $[0, B+S)$ exactly $\gcd(N, K)$ times with $i$ ranging over $[0, N)$ for any $q_{\theta,j}$. Furthermore, for $i\in  Z_j$, we have $Y^{\theta}_{(q_{\theta,j}+i)_{B+S},i}=\mathbf{0}$ by \eqref{Eqn_dummy}, thus together with \eqref{Eqn_ans_j_iteration} we have
\begin{eqnarray*}
A^{(i,\theta)}(j)=\sum\limits_{l=0,l\ne \theta}^{M-1}Y^l_{q_{l,j},i},~\mbox{for~}i\in  Z_j,
\end{eqnarray*}
by which we can get
\begin{eqnarray*}
\sum\limits_{l=0,l\ne \theta}^{M-1}Y^l_{q_{l,j},i},~\mbox{for~}i\in [0, N)\backslash Z_j
\end{eqnarray*}
since \eqref{Eqn_codeword} is a codeword of the $(N, K)$ MDS array code. Therefore, for $j\in [0, S-1)$, by  \eqref{Eqn_ans_j_iteration}, we can now get
\begin{eqnarray*}
Y^{\theta}_{(q_{\theta,j}+i)_{B+S},i},~\mbox{for~}i\in [0, N)\backslash Z_j
\end{eqnarray*}
from
$$\left(A^{(0,\theta)}(j), A^{(1,\theta)}(j), \ldots, A^{(N-1,\theta)}(j)\right).$$
Then we have all the data in the following set with cardinality being $S(N-K)=BK$,
\begin{eqnarray}\label{Eqn_all_coll_data}
\bigcup_{j=0}^{S-1}\{Y^{\theta}_{(q_{\theta,j}+i)_{B+S},i}|i\in [0, N)\backslash Z_j\},
\end{eqnarray}
\textit{i.e.,}
the following data is available,
\begin{eqnarray}\label{Eqn:avail_data}
\setlength{\arraycolsep}{0.5pt}
\left(\hspace{-0.8mm}
  \begin{array}{cccc}
    Y^{\theta}_{q_{\theta,0},0} & Y^{\theta}_{(q_{\theta,0}+1)_{B+S},1} & \cdots & Y^{\theta}_{(q_{\theta,0}+N-1)_{B+S},N-1} \\
     Y^{\theta}_{q_{\theta,1},0} & Y^{\theta}_{(q_{\theta,1}+1)_{B+S},1} & \cdots & Y^{\theta}_{(q_{\theta,1}+N-1)_{B+S},N-1} \\
\vdots & \vdots  & \ddots & \vdots  \\
Y^{\theta}_{q_{\theta,S-1},0} & Y^{\theta}_{(q_{\theta,S-1}+1)_{B+S},1} & \cdots & Y^{\theta}_{(q_{\theta,S-1}+N-1)_{B+S},N-1} \\
  \end{array}
\hspace{-0.8mm}\right)
\end{eqnarray}

For $t\in [0, B)$, according to \eqref{Eqn_encoding}, the $t$-th sub-stripe $\left(W^{\theta}_{t,0}, W^{\theta}_{t,1}, \ldots, W^{\theta}_{t,K-1}\right)$ of file $\theta$ can be reconstructed if we can retrieve any $K$ code-symbols of the following codeword
\begin{equation}\label{Eqn:codew}
\left(Y^{\theta}_{t, 0}, Y^{\theta}_{t, 1}, \ldots, Y^{\theta}_{t, N-1}\right)
\end{equation}
For  $j\in [0, S)$, we  define
\begin{eqnarray*}
  U_{t,j} =  \{i\in [0, N):(q_{\theta,j}+i)_{B+S}=t\},
\end{eqnarray*}
then $U_{t,j}$ implies  which  elements in the $j$-th block row of the matrix in  \eqref{Eqn:avail_data} are code-symbols in \eqref{Eqn:codew}.
Let
\begin{eqnarray*}
U_{t} = \bigcup_{j=0}^{S-1}U_{t,j},
\end{eqnarray*}
then   $U_{t}$ indicates all such elements in the matrix in \eqref{Eqn:avail_data} that are code-symbols in \eqref{Eqn:codew}.

Similarly we have $|U_{t,j}|=\gcd(N, K)$ for any $j\in [0, S)$, therefore, $|U_{t}|=S\times\gcd(N, K)=K$, which implies that we can get $K$ code-symbols in \eqref{Eqn:codew} from the available  data in \eqref{Eqn:avail_data} (or \eqref{Eqn_all_coll_data}).  Finally, we can reconstruct $W^{\theta}_{t,0}, W^{\theta}_{t,1}, \ldots, W^{\theta}_{t,K-1}$ by \eqref{Eqn_encoding}. With $t$ ranging over $[0, B)$, we get the file $W^{\theta}$.

{\bf  Proof of privacy:}
The PIR protocol is private since $Q^{(i,\theta)}$ is uniformly distributed on the set $\Omega$  defined in \eqref{Eqn_Q_i}, regardless of the value of $\theta$.

{\bf  Length of the responses:}
The expected length of the responses is
\begin{eqnarray*}
&&E\left(\sum\limits_{i=0}^{N-1}l_i\right)\\&=&\sum\limits_{i=0}^{N-1}E\left(l_i\right)\\&=&\alpha\sum\limits_{i=0}^{N-1}\sum\limits_{j=0}^{S-1}\mathrm{Prob}(\min \{q_{0,j}, \ldots, q_{\theta-1,j}, (q_{\theta,j}+i)_{B+S},\\&&\hspace{3.5cm} q_{\theta+1,j}, \ldots, q_{M-1,j}\}<B)\\&=&\alpha\sum\limits_{i=0}^{N-1}\sum\limits_{j=0}^{S-1}\left(1-\left(\frac{S}{B+S}\right)^M\right)\\&=&\alpha SN\left(1-\left(\frac{K}{N}\right)^M\right),
\end{eqnarray*}
where the second equality follows from \eqref{Eqn_dummy}, \eqref{Eqn_ans_j_iteration}, and \eqref{Eqn_len},   the third equality follows from Proposition \ref{prop2}.

By \eqref{Eqn_rate}, the retrieval rate is
\begin{eqnarray*}
\frac{L}{E\left(\sum\limits_{i=0}^{N-1}l_i\right)}=\frac{\alpha BK}{\alpha SN(1-(\frac{K}{N})^M)}=\frac{1-\frac{K}{N}}{1-(\frac{K}{N})^M},
\end{eqnarray*}
where the last equality follows from \eqref{Eqn_bs}. Hence, the  rate achieves the capacity in \eqref{Eqn_capacity_MDS} and the capacity of  PIR  from MSR-coded servers that derived in \cite{Patra_capacity_MSR}.
\end{proof}

\begin{Remark}
In contrast to those in \cite{Razane-IT18,Freij-SIAM}, the   retrieval rate of the proposed PIR protocol can achieve the capacity in \eqref{Eqn_capacity_MDS} because  the download
cost is asymmetric in different realizations of queries, which   inherits the techniques in \cite{Tian2019,Zhu2019,Zhou2019}.
\end{Remark}

\subsection{An illustrative example}
In this subsection, we give an illustrative example of a PIR protocol from the first $(N, K, \alpha=(N-K)^N)$ MDS array code with repair degree $D=N-1$ in \cite{Barg1}, where we choose $N=5$, $K=3$, and $M=2$. Then $B=2$ and $S=3$ according to \eqref{Eqn_bs} and $\alpha=32$. The data stored at each server is depicted as in Table \ref{Table_ex_struc}.

\begin{table}[htbp]
\begin{center}
\caption{The data stored at each server for the  PIR protocol based on a $(5, 3, 2^5)$ MDS array code }\label{Table_ex_struc}
\setlength{\tabcolsep}{4.2pt}
\begin{tabular}{|c|c|c|c|c|}
\hline
Server 0  & Server 1 &  Server 2& Server 3 & Server 4\\
\hline
$Y_{0,0}^0$  & $Y_{0,1}^0$ & $Y_{0,2}^0$ &$Y_{0,3}^0$ & $Y_{0,4}^0$\\
\hline
$Y_{1,0}^0$  & $Y_{1,1}^0$ & $Y_{1,2}^0$ &$Y_{1,3}^0$ & $Y_{1,4}^0$\\
\hline
$Y_{0,0}^1$  & $Y_{0,1}^1$ & $Y_{0,2}^1$ &$Y_{0,3}^1$ & $Y_{0,4}^1$\\
\hline
$Y_{1,0}^1$  & $Y_{1,1}^1$ & $Y_{1,2}^1$ &$Y_{1,3}^1$ & $Y_{1,4}^1$\\
\hline
\end{tabular}
\end{center}
\end{table}
Assume that the user requests $W^0$,  and randomly chooses a $2\times 3$ matrix $Q$ from \eqref{Eqn_query_space} as
\begin{equation*}
 Q=\left(\begin{array}{ccc}
 q_{0,0} & q_{0,1}  &  q_{0,2} \\
  q_{1,0} & q_{1,1}  &  q_{1,2} \\
\end{array}\right).
\end{equation*}
By \eqref{Eqn_Q_i}, the user sends the query $Q^{(j,0)}$ to server $j$:
\begin{eqnarray*}
 Q^{(0,0)}&=&\left(\begin{array}{ccc}
  q_{0,0} & q_{0,1}  &  q_{0,2} \\
  q_{1,0} & q_{1,1}  &  q_{1,2} \\
\end{array}\right),\\Q^{(1,0)}&=&\left(\begin{array}{ccc}
 (q_{0,0}+1)_{5} & (q_{0,1}+1)_{5}  &  (q_{0,2}+1)_{5} \\
  q_{1,0} & q_{1,1}  &  q_{1,2} \\
\end{array}\right),\\Q^{(2,0)}&=&\left(\begin{array}{ccc}
 (q_{0,0}+2)_{5} & (q_{0,1}+2)_{5}  &  (q_{0,2}+2)_{5} \\
  q_{1,0} & q_{1,1}  &  q_{1,2} \\
\end{array}\right),\\Q^{(3,0)}&=&\left(\begin{array}{ccc}
  (q_{0,0}+3)_{5} & (q_{0,1}+3)_{5}  &  (q_{0,2}+3)_{5} \\
  q_{1,0} & q_{1,1}  &  q_{1,2} \\
\end{array}\right),\\Q^{(4,0)}&=&\left(\begin{array}{ccc}
 (q_{0,0}+4)_{5} & (q_{0,1}+4)_{5}  &  (q_{0,2}+4)_{5} \\
  q_{1,0} & q_{1,1}  &  q_{1,2} \\
\end{array}\right).
\end{eqnarray*}

Upon receiving the queries, each server responds  according to  \eqref{Eqn_ans_j} and \eqref{Eqn_ans_j_iteration} as:
\begin{eqnarray*}
 A^{(0,0)}&=&\left(Y_{q_{0,0},0}^0+Y_{q_{1,0},0}^1,~ Y_{q_{0,1},0}^0+Y_{q_{1,1},0}^1, \right.\\ &&\hspace{3.3cm}\left. Y_{q_{0,2},0}^0+Y_{q_{1,2},0}^1\right),\\
 A^{(1,0)}&=&\left(Y_{(q_{0,0}+1)_5,1}^0+Y_{q_{1,0},1}^1,~ Y_{(q_{0,1}+1)_5,1}^0+Y_{q_{1,1},1}^1, \right.\\ &&\hspace{3.3cm}\left. Y_{(q_{0,2}+1)_5,1}^0+Y_{q_{1,2},1}^1\right),\\
 A^{(2,0)}&=&\left(Y_{(q_{0,0}+2)_5,2}^0+Y_{q_{1,0},2}^1,~ Y_{(q_{0,1}+2)_5,2}^0+Y_{q_{1,1},2}^1, \right.\\ &&\hspace{3.3cm}\left. Y_{(q_{0,2}+2)_5,2}^0+Y_{q_{1,2},2}^1\right),\\
 A^{(3,0)}&=&\left(Y_{(q_{0,0}+3)_5,3}^0+Y_{q_{1,0},3}^1,~ Y_{(q_{0,1}+3)_5,3}^0+Y_{q_{1,1},3}^1, \right.\\ &&\hspace{3.3cm}\left. Y_{(q_{0,2}+3)_5,3}^0+Y_{q_{1,2},3}^1\right),\\
 A^{(4,0)}&=&\left(Y_{(q_{0,0}+4)_5,4}^0+Y_{q_{1,0},4}^1,~ Y_{(q_{0,1}+4)_5,4}^0+Y_{q_{1,1},4}^1, \right.\\ &&\hspace{3.3cm}\left. Y_{(q_{0,2}+4)_5,4}^0+Y_{q_{1,2},4}^1\right).
\end{eqnarray*}

The user is able to retrieve  file $W^0$ according to the following procedure:

\begin{itemize}
    \item [(i)] From $A^{(i,0)}(0)$, $i\in [0, 5)$, the user gets
\begin{eqnarray}\label{Eqn_ex_ge}
\nonumber &&  \hspace{-1.2cm}\left(\hspace{-1mm}Y_{q_{0,0},0}^0\hspace{-0.3mm}+\hspace{-0.3mm}Y_{q_{1,0},0}^1, Y_{(q_{0,0}+1)_5,1}^0+Y_{q_{1,0},1}^1, Y_{(q_{0,0}+2)_5,2}^0\right.\\
&&\hspace{-1.2cm}\left. \hspace{-0.3mm}+\hspace{-0.3mm}Y_{q_{1,0},2}^1,  Y_{(q_{0,0}+3)_5,3}^0\hspace{-0.3mm}+\hspace{-0.3mm}Y_{q_{1,0},3}^1, Y_{(q_{0,0}+4)_5,4}^0\hspace{-0.3mm}+\hspace{-0.3mm}Y_{q_{1,0},4}^1\hspace{-1mm}\right)\hspace{-1mm}.
\end{eqnarray}
Note that
\begin{eqnarray}\label{Eqn_ex_0_5}
\{q_{0,0},~(q_{0,0}+1)_5,\ldots, (q_{0,0}+4)_5\}=[0,5),
\end{eqnarray}
thus according to \eqref{Eqn_dummy}, exactly three vectors among $$Y_{q_{0,0},0}^0,~Y_{(q_{0,0}+1)_5,1}^0,\ldots,Y_{(q_{0,0}+4)_5,4}^0$$ are $\mathbf{0}_{\alpha}$, \textit{i.e.,} from \eqref{Eqn_ex_ge}, the user can directly obtain  three vectors among
\begin{eqnarray}\label{Eqn_ex_ge_Y}
Y_{q_{1,0},0}^1,~Y_{q_{1,0},1}^1,\ldots,Y_{q_{1,0},4}^1,
\end{eqnarray}
from which the user can further get the rest two vectors in \eqref{Eqn_ex_ge_Y} by \eqref{Eqn_encoding}. With \eqref{Eqn_ex_ge} and \eqref{Eqn_ex_ge_Y} available, the user can get
\begin{equation}\label{Eqn_ex_ve_1}
 \small \hspace{-7mm}\!\left(Y_{q_{0,0},0}^0,  Y_{(q_{0,0}+1)_5,1}^0, Y_{(q_{0,0}+2)_5,2}^0,  Y_{(q_{0,0}+3)_5,3}^0, Y_{(q_{0,0}+4)_5,4}^0\right).\!
\end{equation}

\item [(ii)] From $A^{(i,0)}(1)$, $i\in [0, 5)$, the user can similarly obtain
\begin{equation}\label{Eqn_ex_ve_2}
\small \hspace{-7mm}\left(Y_{q_{0,1},0}^0,~ Y_{(q_{0,1}+1)_5,1}^0,~ Y_{(q_{0,1}+2)_5,2}^0,~  Y_{(q_{0,1}+3)_5,3}^0,~ Y_{(q_{0,1}+4)_5,4}^0\right).
\end{equation}
\item [(iii)] From $A^{(i,0)}(2)$, $i\in [0, 5)$, the user similarly gets
\begin{equation}\label{Eqn_ex_ve_3}
\small \hspace{-7mm}\left(Y_{q_{0,2},0}^0,~ Y_{(q_{0,2}+1)_5,1}^0,~ Y_{(q_{0,2}+2)_5,2}^0,~  Y_{(q_{0,2}+3)_5,3}^0,~ Y_{(q_{0,2}+4)_5,4}^0\right).
\end{equation}
\end{itemize}
Note that exactly two block entries in each of the block vectors in \eqref{Eqn_ex_ve_1}-\eqref{Eqn_ex_ve_3} are not $\mathbf{0}_{\alpha}$ by \eqref{Eqn_dummy} and  \eqref{Eqn_ex_0_5}, which can be denoted as \begin{equation}\label{Eqn_ex_avaY}
Y_{0,u_0}^0,~ Y_{0,u_1}^0,~ Y_{0,u_2}^0,~Y_{1,v_0}^0,~Y_{1,v_1}^0,~Y_{1,v_2}^0,
\end{equation}
for some $u_0,~u_1,~u_2,~v_0,~v_1,~v_2\in [0,5)$, where
$u_0,~u_1,~u_2$ are pairwise distinct and $v_0,~v_1,~v_2$ are also pairwise distinct since $q_{0,0},~q_{0,1},~q_{0,2}$ are pairwise distinct.

With the   data in \eqref{Eqn_ex_avaY} available, $W^0$ can be retrieved according to \eqref{Eqn_encoding}.

From the above analysis, we can derive the length of the responses as
\begin{eqnarray*}
\sum\limits_{i=0}^{4}l_i&=&\alpha\sum\limits_{i=0}^{4}\sum\limits_{j=0}^{2}\mathbf{1}_{\{A^{(i,0)}(j)\ne \mathbf{0}_{\alpha}\}}\\&=&\alpha\sum\limits_{j=0}^{2}\sum\limits_{i=0}^{4}\mathbf{1}_{\{A^{(i,0)}(j)\ne \mathbf{0}_{\alpha}\}}\\&=&\alpha\sum\limits_{j=0}^{2}(2+3\times \mathbf{1}_{\{q_{1,j}<2\}})\\&=&6\alpha+3\alpha\sum\limits_{j=0}^{2}\mathbf{1}_{\{q_{1,j}<2\}}  ,
\end{eqnarray*}
where $\mathbf{1}_{\{*\}}$ is an indicator function defined by
\begin{equation*}
    \mathbf{1}_{\{*\}}=\left\{\begin{array}{cc}
    1, & \mbox{if\,\,*\,\,is\,\,true},\\
 0, & \mbox{otherwise}.
\end{array}\right.
\end{equation*}

Let us go further if we know the matrix $Q$, for example, if
\begin{equation*}
 Q=\left(\begin{array}{ccc}
 0 & 2  &  4 \\
  1 & 3  &  0 \\
\end{array}\right),
\end{equation*}
then the queries are
\begin{eqnarray*}
 Q^{(0,0)}&=&\left(\begin{array}{ccc}
 0 & 2  &  4 \\
  1 & 3  &  0 \\
\end{array}\right),~Q^{(1,0)}=\left(\begin{array}{ccc}
 1 & 3  &  0 \\
  1 & 3  &  0 \\
\end{array}\right),\\Q^{(2,0)}&=&\left(\begin{array}{ccc}
 2 & 4  &  1 \\
  1 & 3  &  0 \\
\end{array}\right),~Q^{(3,0)}=\left(\begin{array}{ccc}
 3 & 0  &  2 \\
  1 & 3  &  0 \\
\end{array}\right),\\Q^{(4,0)}&=&\left(\begin{array}{ccc}
 4 & 1  &  3 \\
  1 & 3  &  0 \\
\end{array}\right).
\end{eqnarray*}

The responses from the servers are:
\begin{eqnarray*}
 A^{(0,0)}&=&(Y_{0,0}^0+Y_{1,0}^1,~ \mathbf{0},~ Y_{0,0}^1),\\
 A^{(1,0)}&=&(Y_{1,1}^0+Y_{1,1}^1,~ \mathbf{0},~ Y_{0,1}^0+Y_{0,1}^1),\\
 A^{(2,0)}&=&(Y_{1,2}^1,~ \mathbf{0},~ Y_{1,2}^0+Y_{0,2}^1),\\
 A^{(3,0)}&=&(Y_{1,3}^1,~ Y_{0,3}^0,~ Y_{0,3}^1),\\
 A^{(4,0)}&=&(Y_{1,4}^1,~ Y_{1,4}^0,~ Y_{0,4}^1).
\end{eqnarray*}

The user is able to retrieve the file $W^0$ according to the following procedure:

\begin{itemize}
    \item [(i)] From $A^{(i,0)}(0)$, $i\in [0, 5)$, the user obtains
\begin{eqnarray}\label{Eqn_ans_0iter}
(Y_{0,0}^0+Y_{1,0}^1,~ Y_{1,1}^0+Y_{1,1}^1,~ Y_{1,2}^1,~  Y_{1,3}^1,~ Y_{1,4}^1).
\end{eqnarray}
From the last three block entries in \eqref{Eqn_ans_0iter}, the user can recover $Y_{1,0}^1$ and $Y_{1,1}^1$ according to \eqref{Eqn_encoding}, which together with the first two block entries in \eqref{Eqn_ans_0iter} give $Y_{0,0}^0$ and $Y_{1,1}^0$.
\item [(ii)] From $A^{(i,0)}(1)$, $i\in [0, 5)$, the user gets $Y_{0,3}^0$ and $Y_{1,4}^0$.
\item [(iii)] From $A^{(i,0)}(2)$, $i\in [0, 5)$, the user obtains
\begin{eqnarray*}
\left(Y_{0,0}^1,~Y_{0,1}^0+Y_{0,1}^1,~Y_{1,2}^0+Y_{0,2}^1,~ Y_{0,3}^1,~Y_{0,4}^1\right),
\end{eqnarray*}
from which the user can similarly obtain $Y_{0,1}^0$ and $Y_{1,2}^0$.
\end{itemize}
With the above available data, $W^0$ can be retrieved according to \eqref{Eqn_encoding}.

Furthermore, in the case of a single server failure, the repair bandwidth is $M\times B\times\gamma_{\mathbf{MSR}}=2^8$ according to \eqref{Eqn_MSR_parameters} since the  code in \cite{Barg1} is an MSR code with repair degree $D=N-1$, which is only a fraction of $\frac{2^8}{M\alpha BK}=\frac{2}{3}$ of the size of all the files.

\section{Applications and Comparisons}

In this section, we first give four new PIR protocols by employing some known MDS array codes to the PIR protocol presented in the previous section, and then give  comparisons of some key metrics among the new   PIR protocols and some existing ones.

\begin{Application}
Suppose we choose the $(N, K, \alpha=K-1)$ PM-MSR code in \cite{PM} to encode the files, where $N
\ge 2K-1$, then we get an  $(N, K)$ PIR protocol,  which is termed new PIR protocol from PM-MSR code.  The sub-packetization of this new PIR protocol is $(K-1)K\frac{N-K}{\gcd(N, K)}$, where each server can be optimally repaired with the repair degree $D$ being $2K-2$, and the retrieval rate meets the capacity in \cite{Banawan_MDS_capacity,Patra_capacity_MSR}. This outperforms the results in \cite{Dorkson_PM_PIR} and \cite{Julien_PM_PIR}. Note that this new PIR protocol is only applicable to the low code rate case.
\end{Application}

\begin{Application}
We can choose a binary MDS array code to encode the files, for example, the  $(N=K+2, K, \alpha=2^{K+1})$ binary MDS array code obtained by operating the transformation in \cite{Bitran} to the code in \cite{MDR2Wang}.  Then we get an  $(N, K)$ PIR protocol over $\mathbf{F}_2$,   which is termed new PIR protocol from binary MDS code in this paper, the sub-packetization is $2^{K+1}K\frac{2}{\gcd(N, K)}$, each server can be optimally repaired with $D=N-1$, and the retrieval rate meets the capacity in \cite{Banawan_MDS_capacity,Patra_capacity_MSR}. The new PIR protocol outperforms most known  PIR protocols from MDS-coded servers in terms of the field size and the repair bandwidth of a  single  server failure, though a  larger sub-packetization is required.
This new PIR protocol is particularly suitable for the high code rate case with rate arbitrarily close to 1.
\end{Application}

\begin{Application}
We  choose the  new $(N, K, \alpha=(N-K)^{\frac{N}{s}-1})$ MDS array code $\mathcal{C}_3$ in \cite{Li_eMSR_ISIT} to encode the files, where $s$ is an arbitrary nontrivial factor of $N$ such that $\frac{N}{s}>N-K$,  the code is over a finite field $\mathbf{F}_q$ with $q>s(N-K)$, and the repair bandwidth is near-optimal. Then we get an  $(N, K)$ PIR protocol over $\mathbf{F}_q$,  which is termed new PIR protocol from $\epsilon$-MSR code, the sub-packetization is $(N-K)^{\frac{N}{s}-1}K\frac{N-K}{\gcd(N, K)}$, each server can be near-optimally repaired, and  the retrieval rate meets the capacity in \cite{Banawan_MDS_capacity,Patra_capacity_MSR}. The new PIR protocol derived here is  particularly suitable for the high code rate case with $\frac{K}{N}>\frac{1}{2}$ since $\frac{N}{s}>N-K$.
\end{Application}

\begin{Application}
We can choose the $(N, K, \alpha=(N-K)^{\frac{N}{N-K}})$ MDS array code obtained from the second application in \cite{transform-IT} to encode the files, this kind of code has the optimal node capacity w.r.t. to the bound in \cite{tight_bound_on_alpha} and is also derived in \cite{Barg2,Sasidharan-Kumar2}, we term the code as optimal node capacity code in this paper.  Then we get an  $(N, K)$ PIR protocol over $\mathbf{F}_q$ with $q>N$, which is termed new PIR protocol from optimal node capacity code, the sub-packetization is $(N-K)^{\frac{N}{N-K}}K\frac{N-K}{\gcd(N, K)}$, each server can be optimally repaired with $D=N-1$, and the retrieval rate meets the capacity in \cite{Banawan_MDS_capacity,Patra_capacity_MSR}.
\end{Application}

\begin{Remark}
In general, there are of course many more applications in addition to the above four, depending on what properties are desired for the system in question. For example, we can also choose the $(N, K)$ MDS array codes in \cite{Barg1} to encode the files, which allow the greatest flexibility in choosing the helper servers when repairing failed servers. Specifically,  the number of helper servers can be anywhere from $K+1$ to $N-1$ and can simultaneously repair multiple server failures. However, the sub-packetization would be larger than those in the above applications.
\end{Remark}

\begin{table*}[htbp]
\begin{center}
\caption{A comparison of some key parameters among the $(N, K)$ PIR protocols proposed in this paper and some existing ones, where the optimal value of the repair bandwidth refers to $\gamma_{MSR}$ in \eqref{Eqn_MSR_parameters} with $D=2K-2$ for the PIR protocols from PM-MSR code and $D=N-1$ for other PIR protocols.}\label{Table comp}
\setlength{\tabcolsep}{1.5pt}
\resizebox{\textwidth}{!}{
\begin{tabular}{|c|c|c|c|c|c|}
\hline
&\multirow{2}{*}{Sub-packatization $L$}& \multirow{2}{*}{Field size $q$ }& The ratio $\overline{\gamma}$ of repair bandwidth   & \multirow{2}{*}{Retrieval rate $R^a$}& \multirow{2}{*}{Constraint} \\
&   &  & to the  optimal value   &  &\\
\hline\hline
New PIR protocol from   & \multirow{2}{*}{$L_1=(K-1)K\frac{N-K}{\gcd(N, K)}$} &  \multirow{2}{*}{$q_1>N$} & \multirow{2}{*}{$\overline{\gamma}_1=1$} &  \multirow{2}{*}{$R^a_1=\frac{1-\frac{K}{N}}{1-(\frac{K}{N})^M}$} & \multirow{2}{*}{$N\ge 2K-1$} \\
 PM-MSR code \cite{PM}&&&&&\\
\hline
PIR protocol from  PM-MSR  & \multirow{2}{*}{$L_5=(K-1)K\frac{N-K}{\gcd(N, K)}$}  & \multirow{2}{*}{$q_5>N$}    & \multirow{2}{*}{$\overline{\gamma}_5=1$}  & \multirow{2}{*}{$R^a_5=1-\frac{2K-2}{N}$} & \multirow{2}{*}{$N\ge 2K-1$}\\
code  by Dorkson \textit{et al.} \cite{Dorkson_PM_PIR}&&&&&\\
\hline
PIR protocol from  PM-MSR  &   \multirow{2}{*}{$L_6=(K-1)K\frac{N-K}{\gcd(N, K)}$}  & \multirow{2}{*}{$q_6>N$}    & \multirow{2}{*}{$\overline{\gamma}_6=1$}  & \multirow{2}{*}{$R^a_6=1-\frac{4K-2}{3N-2K+4}$} & \multirow{2}{*}{$N\ge 2K-1$}\\
code by Lavauzelle \textit{et al.} \cite{Julien_PM_PIR}&&&&&\\
\hline\hline
New PIR protocol from  & \multirow{2}{*}{$L_2=2^{K+1}K\frac{N-K}{\gcd(N, K)}$} &  \multirow{2}{*}{$q_2=2$} & \multirow{2}{*}{$\overline{\gamma}_2=1$} & \multirow{2}{*}{$R^a_2=\frac{1-\frac{K}{N}}{1-(\frac{K}{N})^M}$} & \multirow{2}{*}{$N-K=2$} \\
binary MDS code \cite{Bitran,MDR2Wang} &&&&&\\
\hline
New PIR protocol from  & $L_3=(N-K)^{\frac{N}{s}-1}$ &  \multirow{2}{*}{$q_3>s(N-K)$} &  \multirow{2}{*}{$\overline{\gamma}_3=(1+\frac{(s-1)(N-K-1)}{N-1})$} & \multirow{2}{*}{$R^a_3=\frac{1-\frac{K}{N}}{1-(\frac{K}{N})^M}$} & $\frac{N}{s}>N-K$ \\
$\epsilon$-MSR code \cite{Li_eMSR_ISIT} & $\times K\frac{N-K}{\gcd(N, K)}$&&&& $s\ge 2$\\
\hline
New PIR protocol from optimal   & $L_4=(N-K)^{\frac{N}{N-K}}$ &  \multirow{2}{*}{$q_4>N$} & \multirow{2}{*}{$\overline{\gamma}_4=1$} & \multirow{2}{*}{$R^a_4=\frac{1-\frac{K}{N}}{1-(\frac{K}{N})^M}$} &  \\
node capacity code \cite{transform-IT}&$\times K\frac{N-K}{\gcd(N, K)}$ & &&&\\
\hline
PIR protocols by Zhu \textit{et al.}   & \multirow{2}{*}{$L_7=K\frac{N-K}{\gcd(N, K)}$} &  \multirow{2}{*}{$q_7>N$} & \multirow{2}{*}{$\overline{\gamma}_7=\frac{K(N-K)}{N-1}$} & \multirow{2}{*}{$R^a_7=\frac{1-\frac{K}{N}}{1-(\frac{K}{N})^M}$} &  \\\cite{Zhu2019} and Zhou \textit{et al.} \cite{Zhou2019} &&&&&\\
\hline
PIR protocol by  Banawan   & \multirow{2}{*}{$L_8=KN^M$} &  \multirow{2}{*}{$q_8>N$} & \multirow{2}{*}{$\overline{\gamma}_8=\frac{K(N-K)}{N-1}$} & \multirow{2}{*}{$R^a_8=\frac{1-\frac{K}{N}}{1-(\frac{K}{N})^M}$} &  \\and Ulukus \cite{Banawan_MDS_capacity} &&&&&\\
\hline
\end{tabular}}
\end{center}
\end{table*}

Table \ref{Table comp} gives a comparison of some key parameters among the  $(N, K)$ PIR protocols proposed in this paper and some existing ones \cite{Banawan_MDS_capacity,Zhu2019,Zhou2019,Dorkson_PM_PIR,Julien_PM_PIR}.
It is seen that the retrieval rate of each of the new  PIR protocols achieves the capacity  in \cite{Banawan_MDS_capacity,Patra_capacity_MSR}. Besides, under the same parameters $N$ and $K$, we have the following results.
\begin{itemize}
\item  [i)] For all the new PIR protocols, in the case of a single server failure, it can be (near-)optimally repaired, which outperform the PIR protocols in \cite{Banawan_MDS_capacity,Zhu2019,Zhou2019}, more specifically,
\begin{eqnarray*}
1=\overline{\gamma_1}=\overline{\gamma_2}=\overline{\gamma_4}=\overline{\gamma_5}=\overline{\gamma_6}<\overline{\gamma_3}<\overline{\gamma_7}=\overline{\gamma_8}.
\end{eqnarray*}

\item [ii)] The new PIR protocol from  PM-MSR code has larger retrieval rate than that of the PIR protocol from  PM-MSR code  in \cite{Julien_PM_PIR} and the one in \cite{Dorkson_PM_PIR}, while all the other parameters are the same. It also achieves the capacity in \cite{Banawan_MDS_capacity,Patra_capacity_MSR}. In general, we have
\begin{eqnarray*}
C_{MDS}\hspace{-0.15mm}=\hspace{-0.15mm}R^a_1\hspace{-0.15mm}=\hspace{-0.15mm}R^a_2\hspace{-0.15mm}=\hspace{-0.15mm}R^a_3\hspace{-0.15mm}=\hspace{-0.15mm}R^a_4\hspace{-0.15mm}=\hspace{-0.15mm}R^a_7\hspace{-0.15mm}=\hspace{-0.15mm}R^a_8\hspace{-0.15mm}>\hspace{-0.15mm}R^a_6\hspace{-0.15mm}>\hspace{-0.15mm}R^a_5,
\end{eqnarray*}
where $C_{MDS}$ is defined as in \eqref{Eqn_capacity_MDS}.

\item  [iii)] The new PIR protocol from binary MDS code works over $\mathbf{F}_2$, which greatly reduces the complexity of the system since only XOR operations are needed, although the protocol only works for two parity servers. Nevertheless, few servers can be considered beneficial from the collusion and network congestion point of view. In addition, the new PIR protocol from $\epsilon$-MSR code is also built on a smaller finite field than that of the existing ones. Clearly, we have
\begin{eqnarray*}
q_2<q_1=q_5=q_6=q_7=q_8=q_4,
\end{eqnarray*}
and additionally $q_2<q_3\le q_1$
if $s<\frac{N}{N-K}$.


\item  [iv)]  In addition to the repair efficiency, the new PIR protocols outperform the one proposed by Banawan and Ulukus \cite{Banawan_MDS_capacity} in terms of the sub-packetization if $M\ge N$. More specifically,  if $M\ge N$, we can derive
\begin{eqnarray*}
\hspace{-3mm}\left\{\begin{array}{ll}
L_7<\hspace{-1mm}L_3\hspace{-1mm}\le \hspace{-1mm}L_4\hspace{-1mm}<L_2\hspace{-1mm}< \hspace{-1mm}L_8, \hspace{-2mm} &\mbox{if}\hspace{1mm}N-K=2,  2\le s< \frac{N}{2},\\
L_7\hspace{-1mm}<\hspace{-1mm}L_a\hspace{-1mm}<\hspace{-1mm}L_8,\hspace{-2mm} &\mbox{if}\hspace{1mm}2<N-K<K-1,\\
L_7\hspace{-1mm}<\hspace{-1mm}L_1\hspace{-1mm}=\hspace{-1mm}L_5\hspace{-1mm}=\hspace{-1mm}L_6\hspace{-1mm}<\hspace{-1mm}L_4\hspace{-1mm}<\hspace{-1mm}L_8,\hspace{-2mm} &\mbox{if}\hspace{1mm}N-K\ge K-1.
\end{array}\right.
\end{eqnarray*}
where $a=3,4$. Additionally,
\begin{eqnarray*}
\hspace{-3mm}\left\{\begin{array}{ll}
\hspace{-1mm} L_4 \hspace{-1mm}\le \hspace{-1mm}L_3, \hspace{-1mm}  &\mbox{if} \hspace{1mm} 2\hspace{-1mm}\le \hspace{-1mm}s\hspace{-1mm}\le \hspace{-1mm} \min\{\frac{N(N-K)}{2N-K}, \frac{K}{N-K}\},\\
\hspace{-1mm} L_3 \hspace{-1mm}<\hspace{-1mm} L_4,\hspace{-1mm} &\mbox{if} \hspace{1mm}  2\hspace{-1mm}\le \hspace{-1mm}N\hspace{-1mm}-\hspace{-1mm}K\hspace{-1mm}<\hspace{-1mm}1\hspace{-1mm}+\hspace{-1mm}\sqrt{K+1}, \frac{N(N-K)}{2N-K}\hspace{-1mm}<\hspace{-1mm}s<\hspace{-1mm}\frac{N}{N-K}.\\
\end{array}\right.
\end{eqnarray*}
\end{itemize}

In particular, under some specific parameters $N, K$, and $M$, we give  the detailed  comparisons of the sub-packetization and the repair bandwidth among the   PIR protocols proposed in this paper and some existing ones \cite{Banawan_MDS_capacity,Zhu2019,Zhou2019,Dorkson_PM_PIR,Julien_PM_PIR}  in Figures \ref{picture 1}-\ref{picture 4} and Figures \ref{picture ga 1}-\ref{picture ga 4}, respectively,
where Figures \ref{picture 1} and \ref{picture ga 1} focus on  $N-K=2$, Figures  \ref{picture 2}  and \ref{picture ga 2} focus on $N-K=3$, Figures \ref{picture 3} and \ref{picture ga 3} focus on fixed code rate $\frac{K}{N}=\frac{1}{2}$, Figures \ref{picture 4} and \ref{picture ga 4} focus  on fixed code rate $\frac{K}{N}=\frac{2}{3}$.
\begin{figure}[htbp]
\centering
\begin{minipage}[t]{0.48\textwidth}
\centering
\includegraphics[scale=0.55]{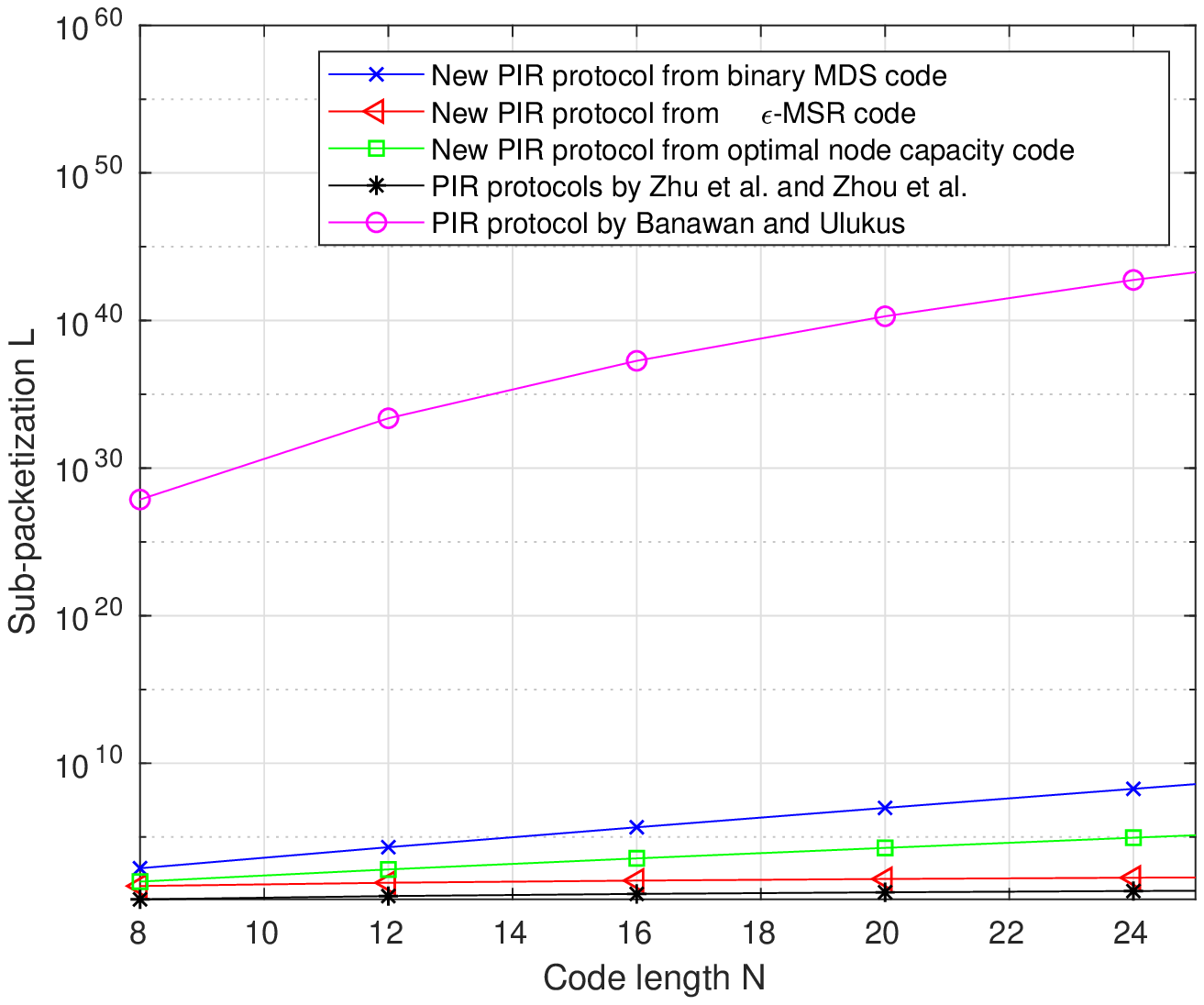}
\caption{Comparisons of the sub-packetization among new PIR protocols and some known ones under the parameters $N-K=2, s=\frac{N}{4}$, and $M=30$}\label{picture 1}
\end{minipage}
\hfill
\begin{minipage}[t]{0.48\textwidth}
\centering
\includegraphics[scale=0.55]{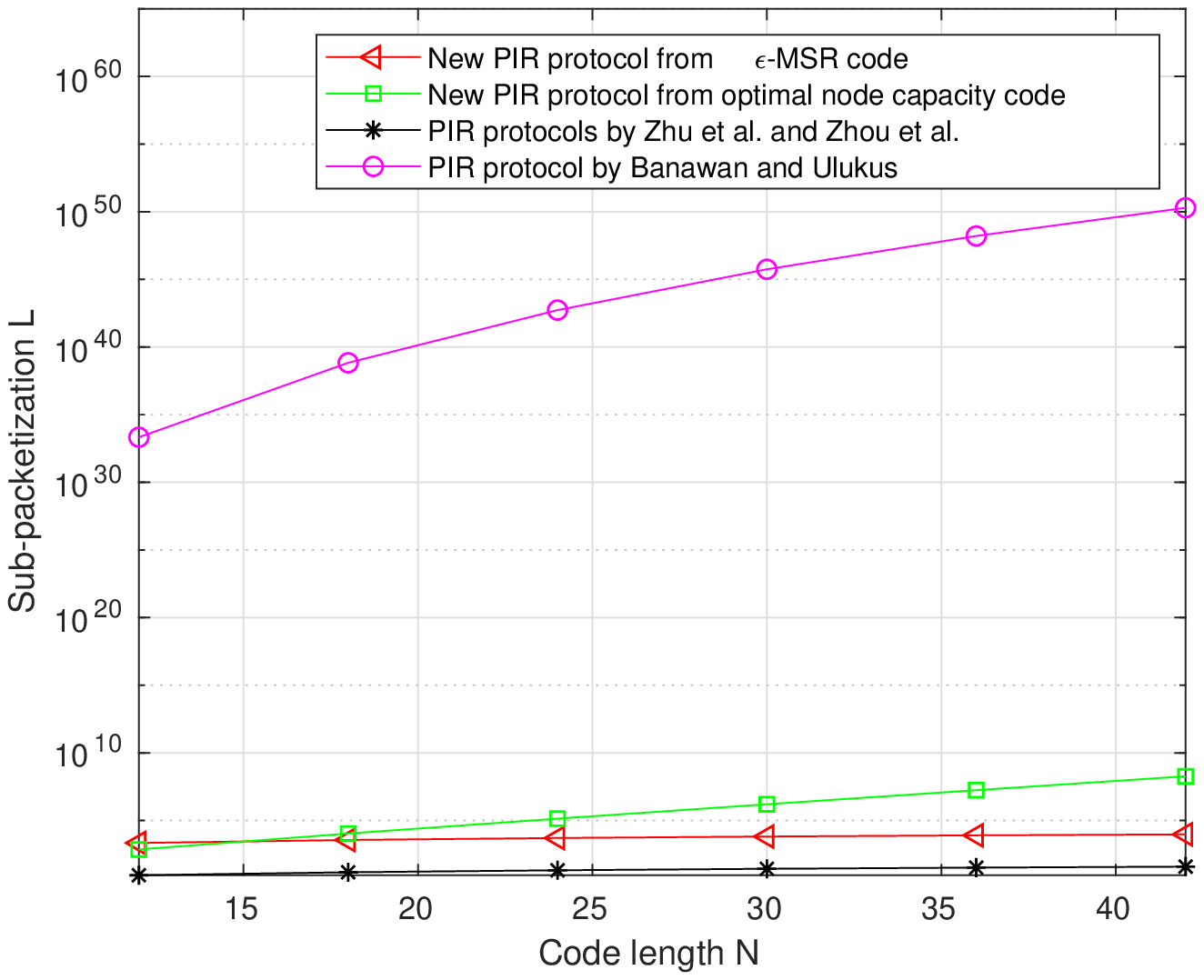}
\caption{Comparisons of the sub-packetization among new PIR protocols and some known ones under the parameters $N-K=3, s=\frac{N}{6}$, and $M=30$}\label{picture 2}
\end{minipage}
\end{figure}

\begin{figure}[htbp]
\centering
\begin{minipage}[t]{0.48\textwidth}
\centering
\includegraphics[scale=0.6]{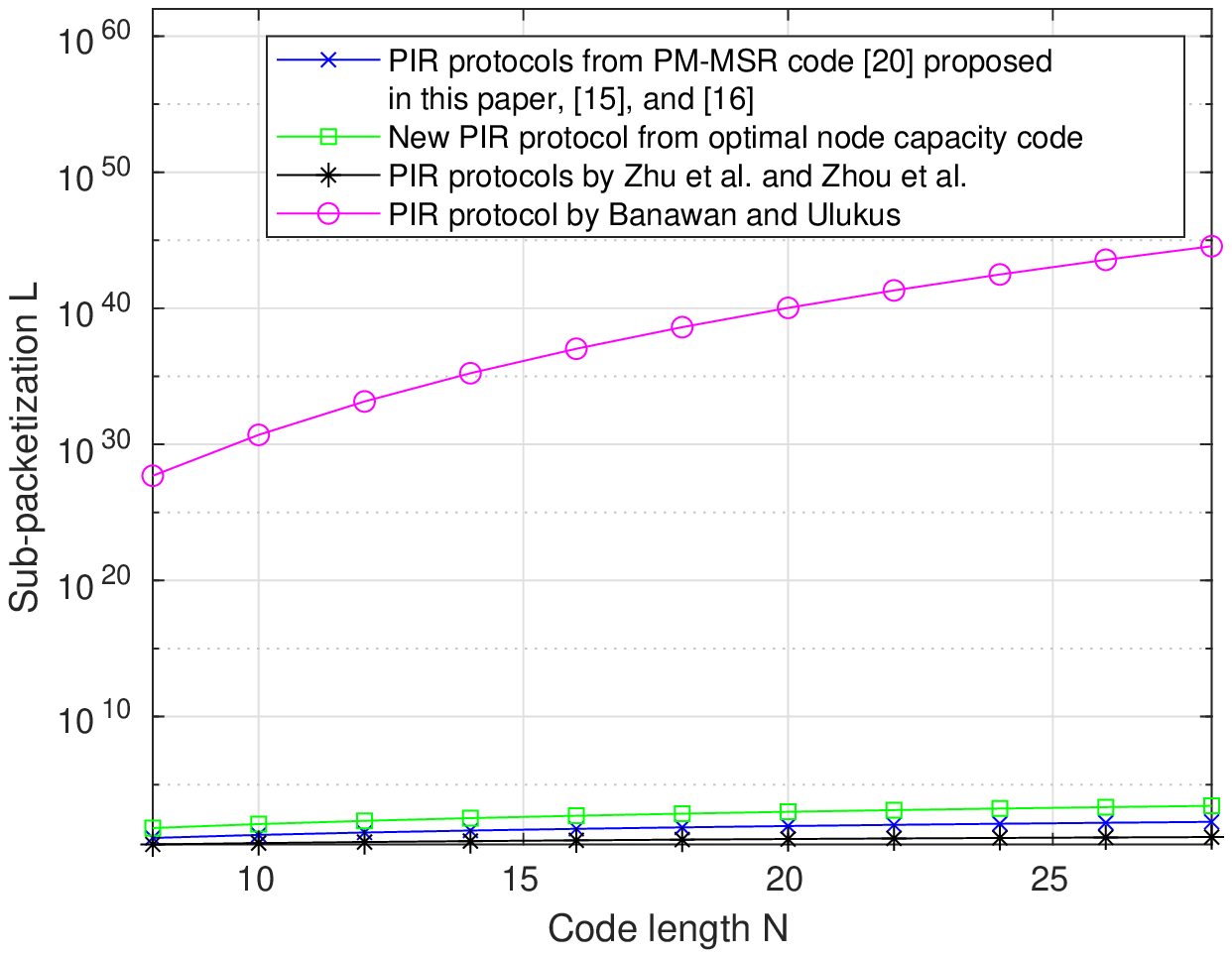}
\caption{Comparisons of the sub-packetization among new PIR protocols and some known ones under fixed code rate $\frac{K}{N}=\frac{1}{2}$, and $M=30$}\label{picture 3}
\end{minipage}
\hfill
\begin{minipage}[t]{0.48\textwidth}
\centering
\includegraphics[scale=0.55]{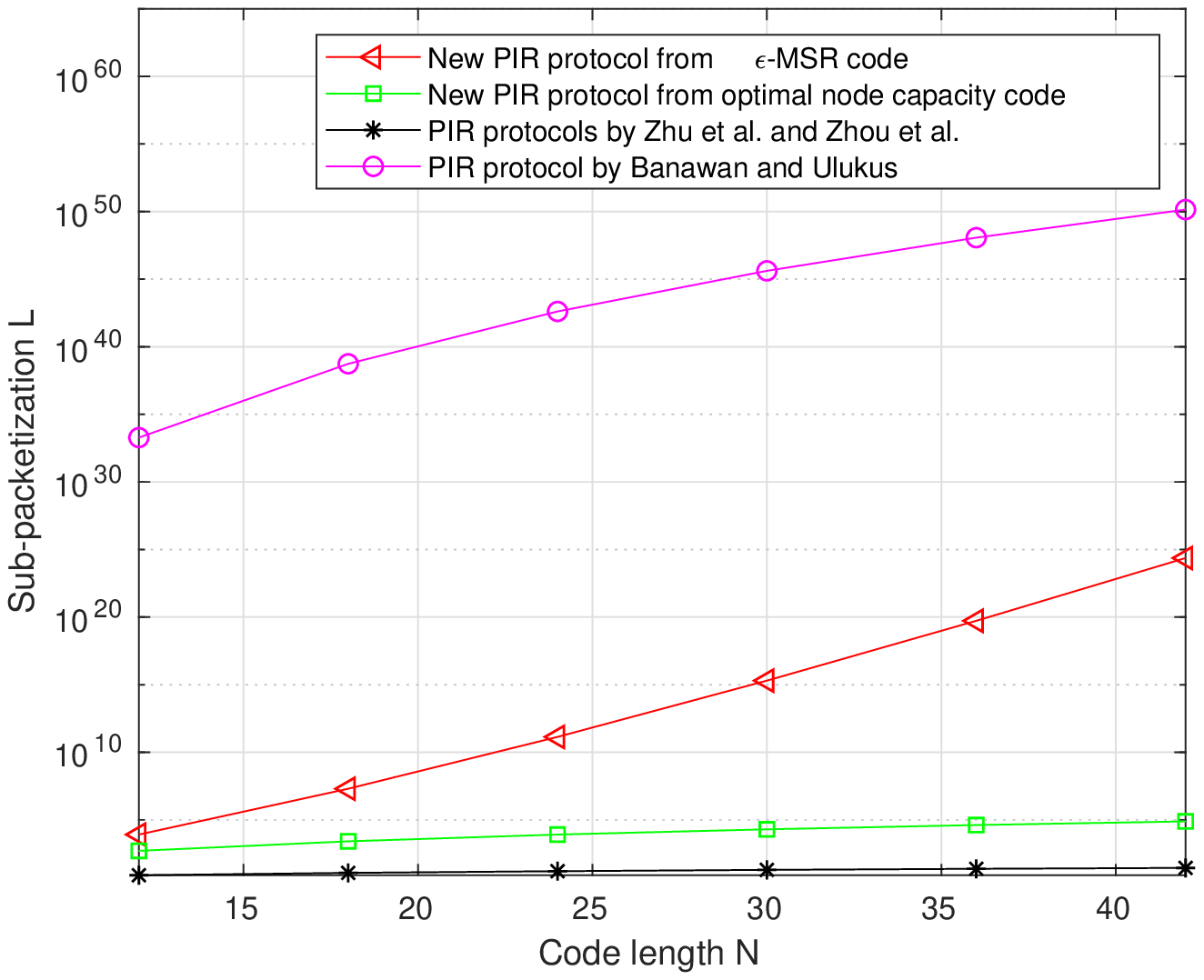}
\caption{Comparisons of the sub-packetization among new PIR protocols and some known ones under fixed code rate $\frac{K}{N}=\frac{2}{3}$, and $M=30$}\label{picture 4}
\end{minipage}
\end{figure}

\begin{figure}[htbp]
\centering
\begin{minipage}[t]{0.48\textwidth}
\centering
\includegraphics[scale=0.55]{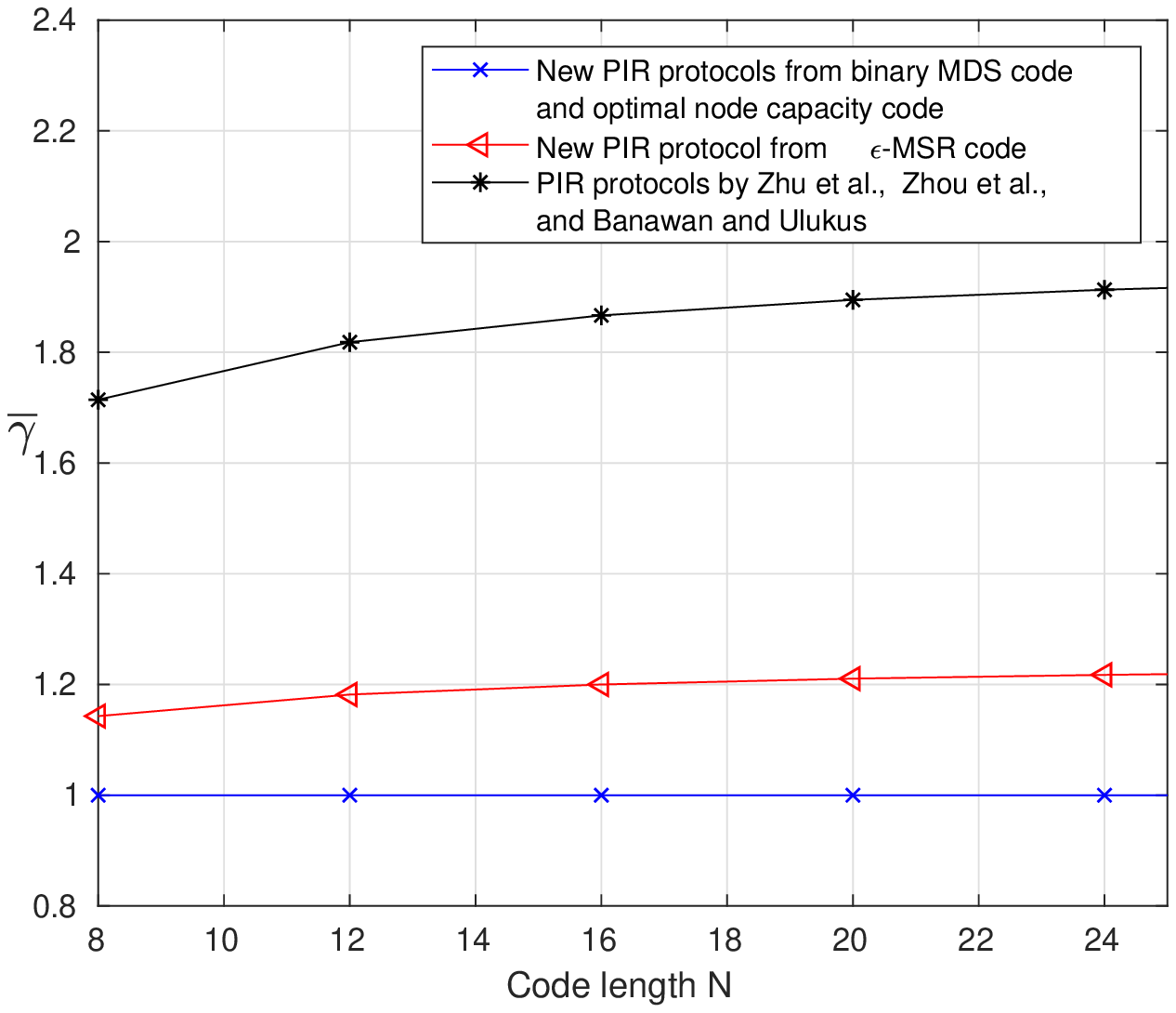}
\caption{Comparisons of the ratio $\overline{\gamma}$ of repair bandwidth  to the  optimal value   among new PIR protocols and some known ones under the parameters $N-K=2, s=\frac{N}{4}$}\label{picture ga 1}
\end{minipage}
\hfill
\begin{minipage}[t]{0.48\textwidth}
\centering
\includegraphics[scale=0.55]{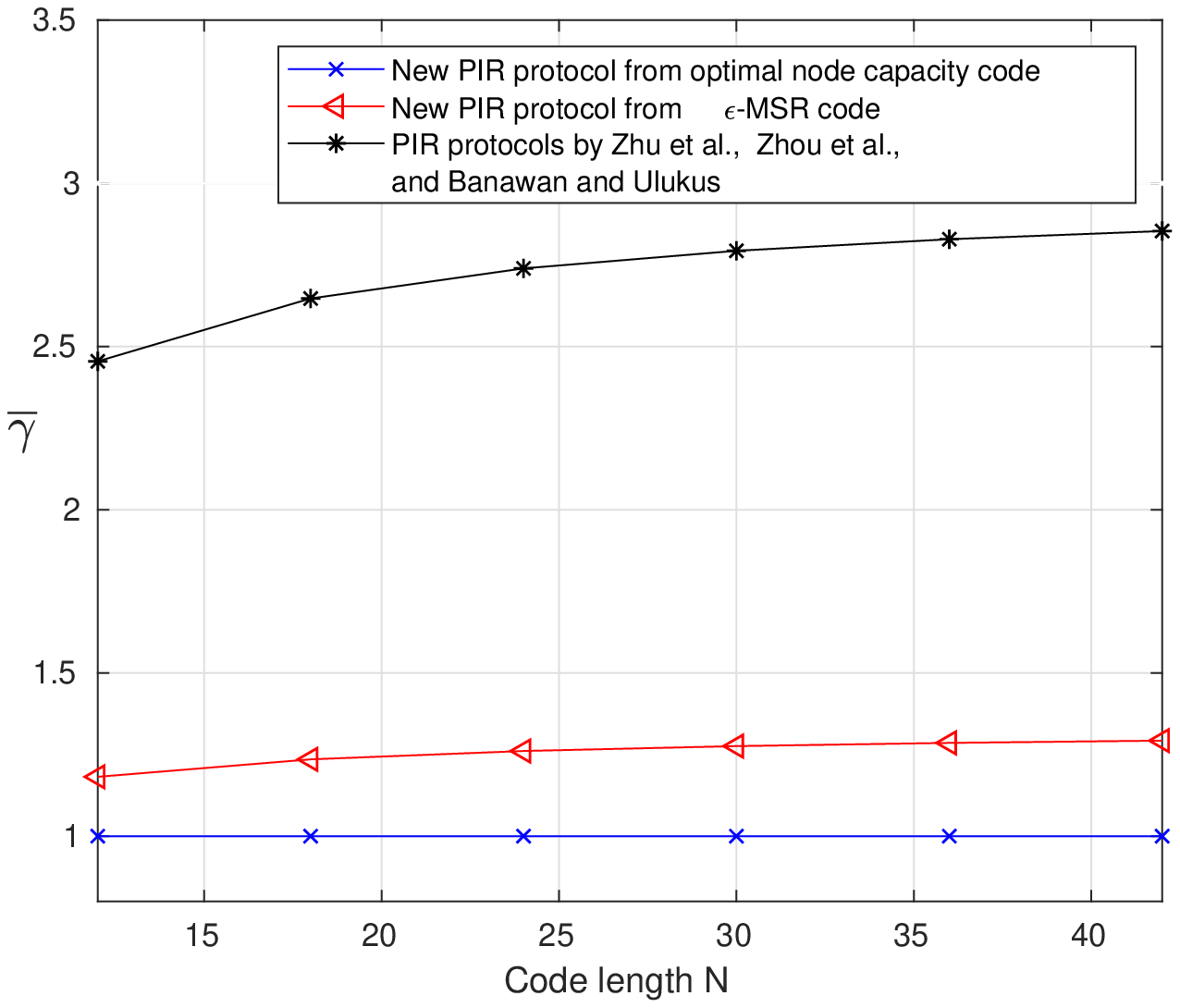}
\caption{Comparisons of the ratio $\overline{\gamma}$ of repair bandwidth  to the  optimal value among new PIR protocols and some known ones under the parameters $N-K=3, s=\frac{N}{6}$}\label{picture ga 2}
\end{minipage}
\end{figure}

\begin{figure}[htbp]
\centering
\begin{minipage}[t]{0.48\textwidth}
\centering
\includegraphics[scale=0.6]{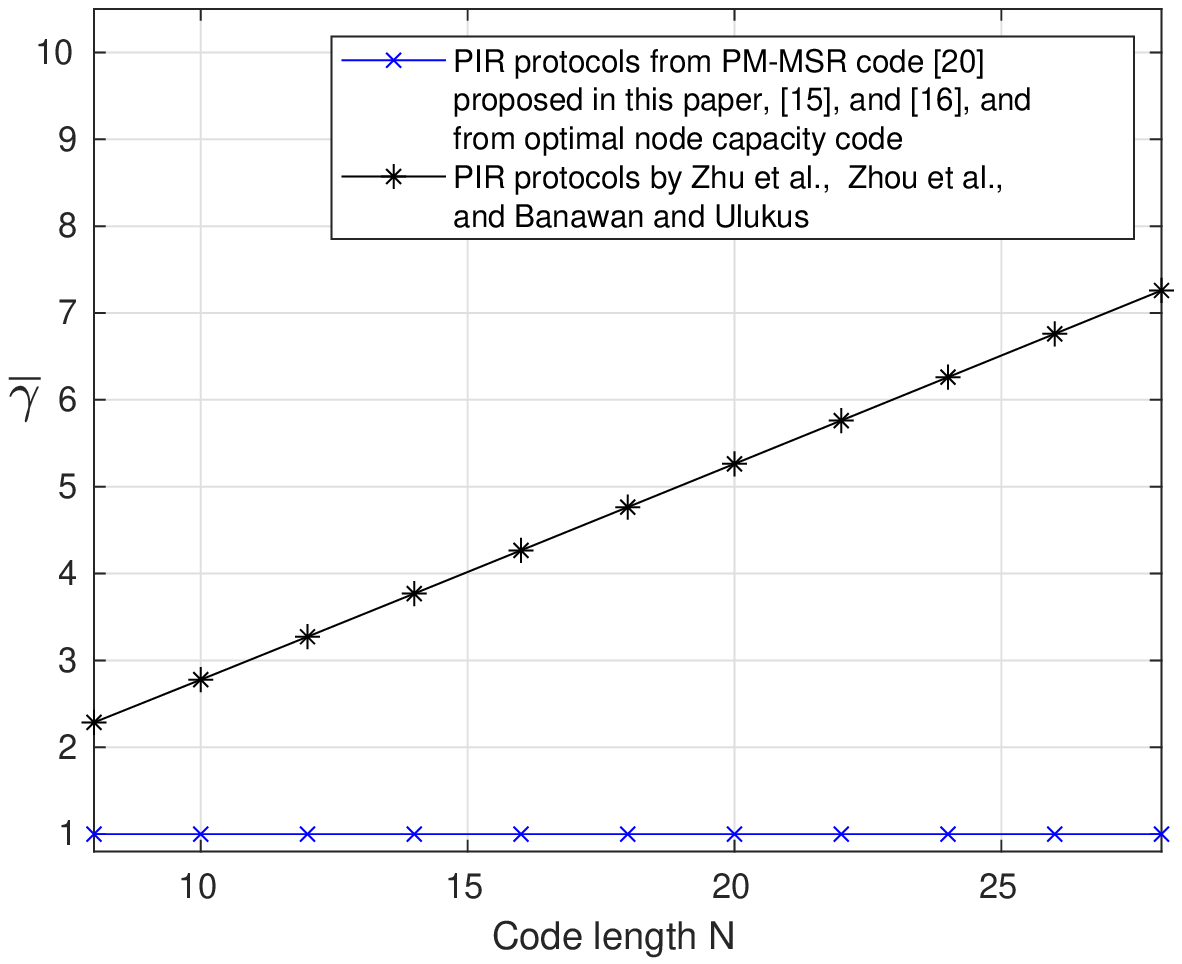}
\caption{Comparisons of the ratio $\overline{\gamma}$ of repair bandwidth  to the  optimal value  among new PIR protocols and some known ones under fixed code rate $\frac{K}{N}=\frac{1}{2}$}\label{picture ga 3}
\end{minipage}
\hfill
\begin{minipage}[t]{0.48\textwidth}
\centering
\includegraphics[scale=0.55]{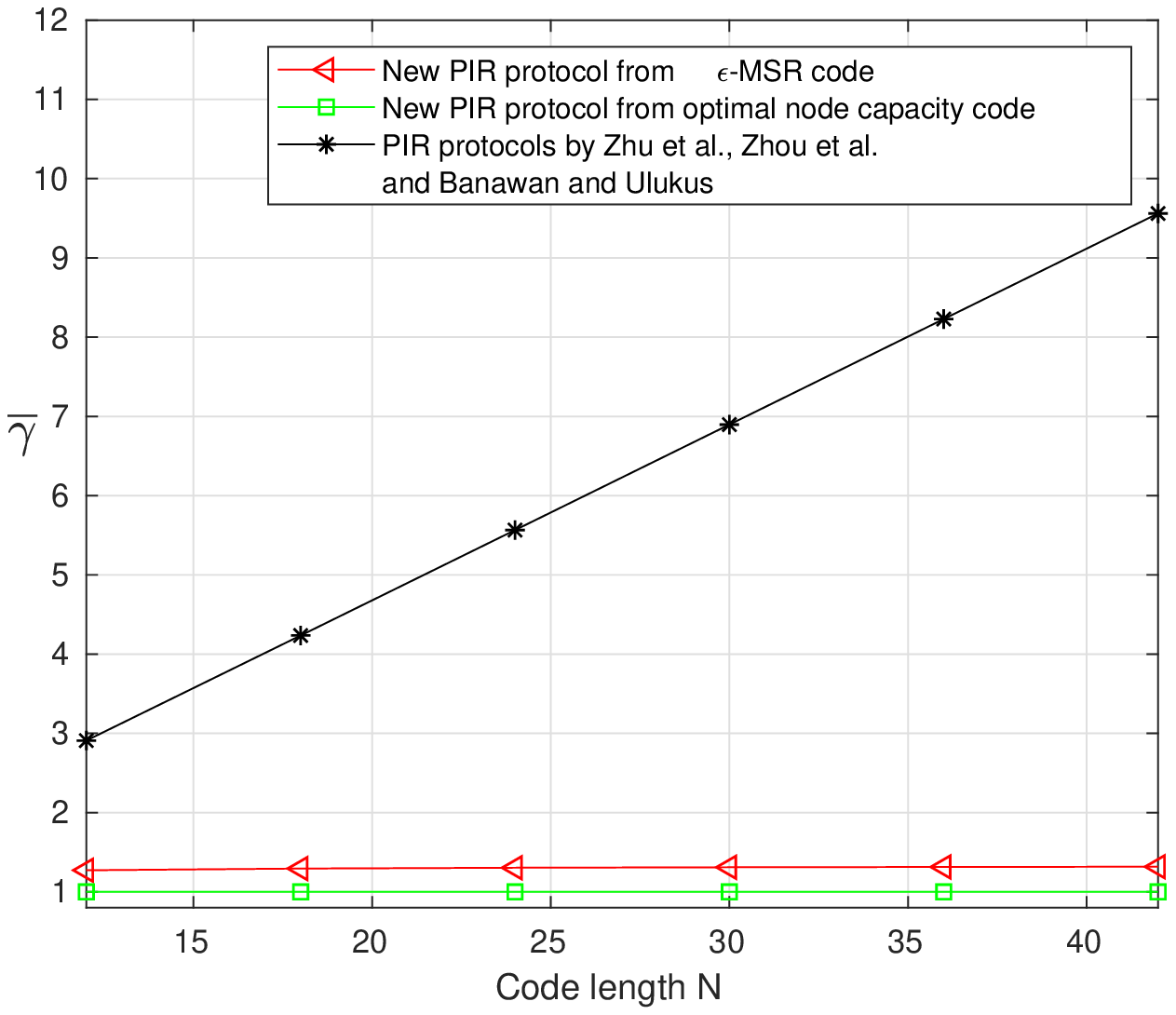}
\caption{Comparisons of the ratio $\overline{\gamma}$ of repair bandwidth  to the  optimal value among new PIR protocols and some known ones under fixed code rate $\frac{K}{N}=\frac{2}{3}$}\label{picture ga 4}
\end{minipage}
\end{figure}

These figures
convince  the previous arguments. From Figures \ref{picture 1}-\ref{picture 4}, we know that the PIR protocols proposed by Zhu \textit{et al.} and Zhou \textit{et al.} are the best among the PIR protocols considered in this paper in terms of the sub-packetization, however, the repair bandwidth is far from optimal according to  Figures \ref{picture ga 1}-\ref{picture ga 4}.
When taking into account both the sub-packetization and the repair bandwidth,   compared to other existing PIR protocols from MDS array codes, we have the following observations:
\begin{itemize}
    \item [i)] According to Figures \ref{picture 1}, \ref{picture 2}, \ref{picture ga 1}, and \ref{picture ga 2}, the new PIR protocol from $\epsilon$-MSR code is very competitive  when the code rate $\frac{K}{N}$ asymptotically close to 1, as it has smaller sub-packetization than all the other PIR protocols except the ones proposed by Zhu \textit{et al.} and \textit{Zhou et al.}, but it has much smaller repair bandwidth than those proposed by Zhu \textit{et al.} and \textit{Zhou et al.}
    \item [ii)] According to Figures \ref{picture 3}  and \ref{picture ga 3}, the new PIR protocol from PM-MSR code is very competitive when the code rate $\frac{K}{N}$ is $\frac{1}{2}$, with the similar arguments stated in i).
       \item [iii)] According to Figures \ref{picture 4}  and \ref{picture ga 4}, the new PIR protocol from optimal node capacity code is very competitive when the code rate $\frac{K}{N}$ is $\frac{2}{3}$, with the similar arguments stated in i).
\end{itemize}

\begin{Remark}
From Figures \ref{picture 1}-\ref{picture ga 4}, it seems that the PIR protocol proposed by Banawan and Ulukus \cite{Banawan_MDS_capacity} is the ``worst" in terms of the sub-packetization and repair bandwidth. However, we would like to clarify that   \cite{Banawan_MDS_capacity} is a pioneer work in PIR from MDS-coded servers, the main aims of which are to determine the capacity of PIR from MDS-coded servers and to find a capacity-achieving protocol, which is the first capacity-achieving result in PIR from MDS-coded servers. The sub-packetization and repair bandwidth were not the aims in \cite{Banawan_MDS_capacity}.
\end{Remark}

\section{Concluding Remarks}
In this paper, we proposed a PIR protocol from MDS array codes with (near-)optimal repair bandwidth, which subsumes PIR from MSR-coded servers as a special case.  The   retrieval rate of the new PIR protocol achieves the capacity of PIR from MDS-/MSR-coded servers. Particularly, four new PIR protocols were obtained by employing some known MDS array codes, with one of them implementable over $\mathbf{F}_2$. In addition to the capacity-achieving rate,  these new PIR protocols have several advantages when compared with existing PIR protocols, such as (near-)optimal repair bandwidth and/or small sub-packetization.
Extending the result to  other cases such as PIR from MDS array codes with colluding servers will be part of our future work.

\section*{Acknowledgment}
The authors would like to thank the Associate Editor Prof. Li Chen and the three anonymous reviewers for their valuable suggestions and comments, which have greatly improved the presentation and quality of this paper. They are also grateful to Prof. Chao Tian and Dr. Jinbao Zhu for helpful discussions.

\ifCLASSOPTIONcaptionsoff
  \newpage
\fi

%
\begin{IEEEbiography}[{\includegraphics[width=1in,height=1.25in,clip,keepaspectratio]{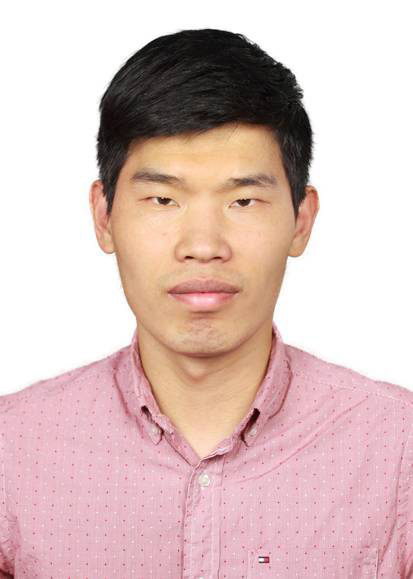}}]
{Jie Li}(S'16-M'17)
received the B.S. and M.S. degrees in mathematics from Hubei University, Wuhan, China, in 2009 and 2012, respectively, and received the Ph.D. degree from the department of communication engineering, Southwest Jiaotong University, Chengdu, China, in 2017. From  2015 to   2016, he was a visiting Ph.D. student in the Department of Electrical Engineering and Computer Science, The University of Tennessee at Knoxville, TN, USA.  From   2017 to   2019, he was  a postdoctoral researcher at the Department of Mathematics, Hubei University, Wuhan, China. Since   2019, he has been a postdoctoral researcher at  the Department of Mathematics and Systems Analysis, Aalto University, Finland. His research interests include private information retrieval, coding for distributed storage, and sequence design.

Dr. Li received the IEEE Jack Keil Wolf ISIT Student Paper Award in 2017.
\end{IEEEbiography}

\begin{IEEEbiography}[{\includegraphics[width=1in,height=1.25in,clip,keepaspectratio]{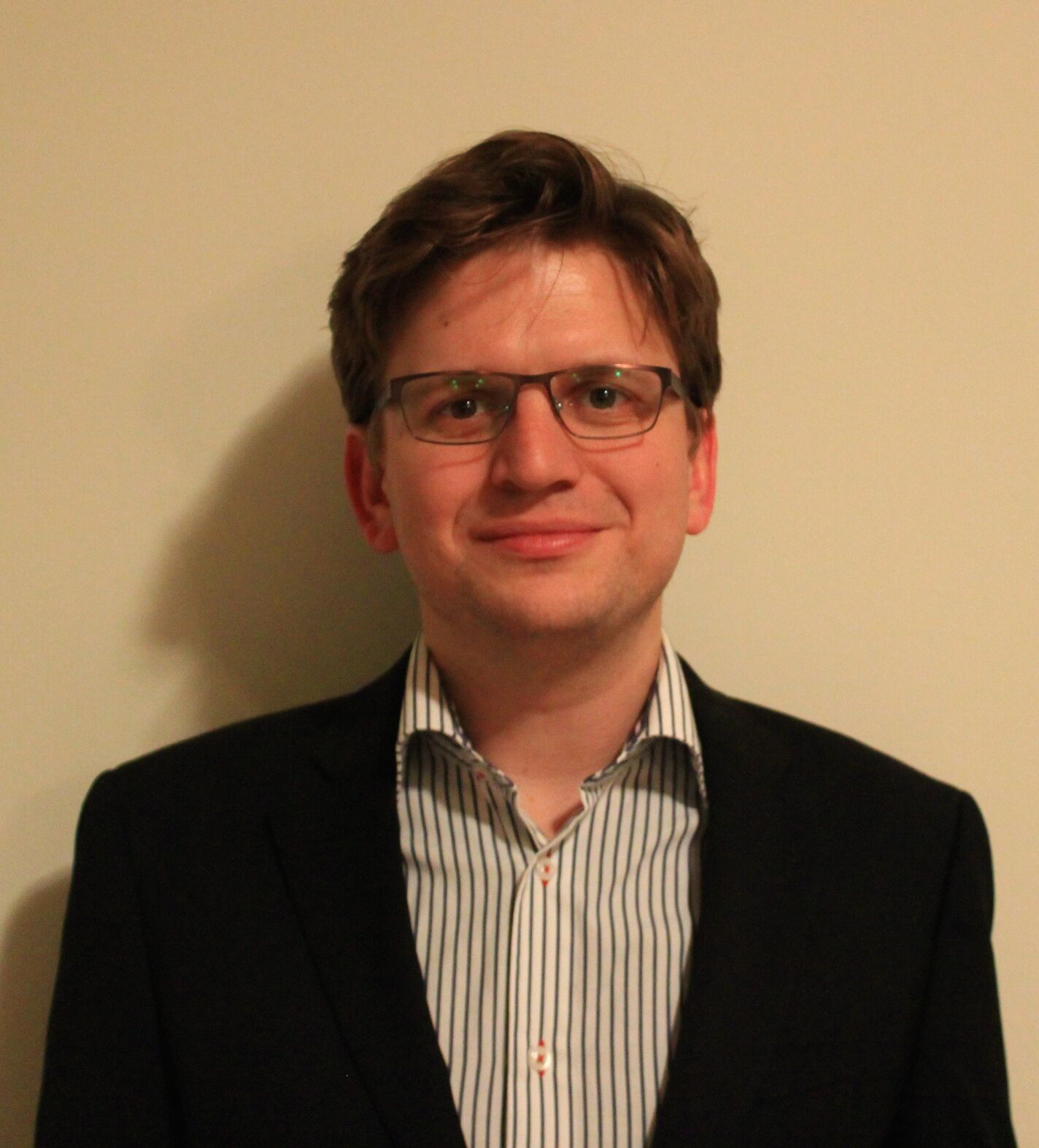}}]
{David Karpuk}(B.A. '06, Ph.D. '12)     received the bachelor's degree in mathematics from Boston College, and the Ph.D. degree in mathematics from University of Maryland, College Park.  He was Postdoctoral Researcher in the Department of Mathematics and Systems Analysis at Aalto University, Helsinki, Finland, from 2012-2017, and an Assistant Professor in the Department of Mathematics at Universidad de los Andes, Bogot\'a, Colombia from 2017-2019.  He was the recipient of Postdoctoral Researcher grants from the Academy of Finland and the Magnus Ehrnrooth Foundation.  Currently he is a Senior Data Scientist at F-Secure Corporation, Helsinki, Finland, where his research interests include applications of machine learning to cyber security.
\end{IEEEbiography}

\begin{IEEEbiography}[{\includegraphics[width=1in,height=1.25in,clip,keepaspectratio]{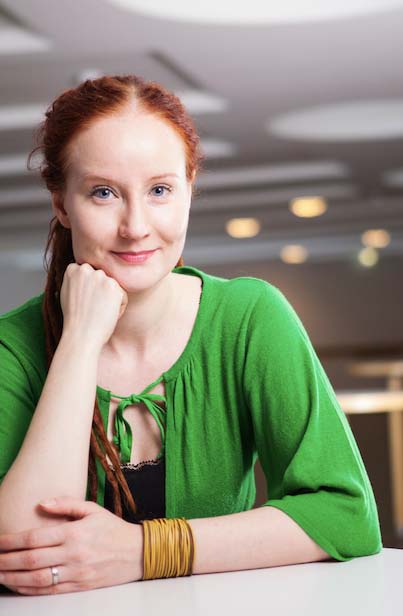}}]
{Camilla Hollanti}(M'09)     received the M.Sc. and Ph.D. degrees from the University of Turku, Finland, in 2003 and 2009, respectively, both in pure mathematics. Her research interests lie within applications of algebraic number theory to wireless communications and physical layer security, as well as in combinatorial and coding theoretic methods related to distributed storage systems and private information retrieval.

For 2004-2011 Hollanti was with the University of Turku. She joined the University of Tampere as  Lecturer for the academic year 2009-2010. Since 2011, she has been with the Department of Mathematics and Systems Analysis at Aalto University, Finland, where she currently works as Full Professor and Vice Head, and leads a research group in Algebra, Number Theory, and Applications. During 2017-2020, Hollanti is also affiliated with the Institute of Advanced Studies at the Technical University of Munich, where she holds a three-year Hans Fischer Fellowship, funded by the German Excellence Initiative and the EU 7th Framework Programme.

Hollanti is an editor of the AIMS Journal on Advances in Mathematics of Communications. She is a recipient of several grants, including five Academy of Finland grants. In 2014, she received the World Cultural Council Special Recognition Award for young researchers. In 2017, the Finnish Academy of Science and Letters awarded her the V\"ais\"al\"a Prize in Mathematics. For 2020-2022, Hollanti will serve as a member of the Board of Governors of the IEEE Information Theory Society, and she is one of the General Chairs of IEEE ISIT 2022.
\end{IEEEbiography}




\end{document}